\documentclass[a4paper,UKenglish]{lipics-v2016}

\usepackage{microtype}\usepackage{booktabs}
\usepackage{xspace}
\usepackage{tikz}
\usetikzlibrary{arrows,shapes,decorations,automata,backgrounds,fit,matrix, decorations.pathreplacing,positioning, quotes, math}
\usepackage{thm-restate}
\usepackage{stmaryrd} 
\usepackage[capitalise]{cleveref} \usepackage{paralist}

\title{Model-checking Counting Temporal Logics on Flat Structures\footnote{\scriptsize Partially supported by EGIDE/DAAD-Procope (FREQS) and European Commission H2020 Project COEMS.
This is an extended version of \cite{DBLP:conf/concur/DeckerHLST17} providing additional \cref{apx:completeness,apx:correctness}.}}

\author[1]{Normann Decker}
\author[2]{Peter Habermehl}
\author[1]{Martin Leucker}
\author[2]{Arnaud Sangnier}
\author[1]{Daniel Thoma}
\affil[1]{ISP, University of Lübeck, Lübeck, Germany\\
  \texttt{\{decker, leucker, thoma\}@isp.uni-luebeck.de}}
\affil[2]{IRIF, Univ.\ Paris Diderot, Paris, France\\
  \texttt{\{habermehl, sangnier\}@irif.fr}\\}

\authorrunning{N. Decker, P. Habermehl, M. Leucker, A. Sangnier, D. Thoma}

\Copyright{Normann Decker, Peter Habermehl, Martin Leucker, Arnaud Sangnier and Daniel Thoma}


\EventEditors{}
\EventNoEds{0}
\EventLongTitle{}
\EventShortTitle{}
\EventAcronym{}
\EventYear{}
\EventDate{}
\EventLocation{}
\EventLogo{}
\SeriesVolume{}
\ArticleNo{~}

\hypersetup{
  pdftitle = {
    Model-checking Counting Temporal Logics on Flat Structures
  },
  pdfauthor = {
    Normann Decker,
    Peter Habermehl,
    Martin Leucker,
    Arnaud Sangnier,
    Daniel Thoma
  }
}

\newcommand{\lipicsbf}[1]{\textcolor{darkgray}{\textsf{\textbf{#1}}}}

\newcommand{\nat}{\mathbb{N}}

\newcommand{\interval}[2]{[#1,#2]}

\newcommand{\set}[1]{\{ #1 \}}
\newcommand{\card}[1]{|#1|}

\newcommand{\vect}[1]{\mathbf{#1}}
\renewcommand{\vec}[1]{\vect{#1}}

\newcommand{\jhat}{\hat{\jmath}}

\newcommand{\length}[1]{|#1|}

\newcommand{\States}{S}
\newcommand{\state}{s}
\newcommand{\Edges}{E}

\newcommand{\KS}{\texttt{KS}\xspace}
\newcommand{\FKS}{\texttt{FKS}\xspace}

\newcommand{\CS}{\texttt{CS}\xspace}

\newcommand{\run}{\rho}
\newcommand{\Runs}{\mathtt{Runs}}

\newcommand{\CTLstar}{\texttt{CTL\textsuperscript{\!*}}\xspace}
\newcommand{\CTL}{\texttt{CTL}\xspace}
\newcommand{\CCTL}{\texttt{CCTL}\xspace}
\newcommand{\CC}{\texttt{$_\mathtt{c}$CTL}\xspace}
\newcommand{\CCpm}{\texttt{$_\mathtt{c}$CTL$_\pm$}\xspace}
\newcommand{\CCstarpm}{\ensuremath{_{\texttt{c}}\texttt{CTL}^{\texttt{*}}_{\pm}}\xspace} 
\newcommand{\CL}{\texttt{$_\mathtt{c}$LTL}\xspace}
\newcommand{\CLpm}{\texttt{$_\mathtt{c}$LTL$_\pm$}\xspace}
\newcommand{\LTL}{\texttt{LTL}\xspace}
\newcommand{\FLTL}{\texttt{fLTL}\xspace}
\newcommand{\fLTL}{\FLTL}
\newcommand{\fCTL}{\texttt{fCTL}\xspace}
\newcommand{\fCTLstar}{\texttt{fCTL\textsuperscript{\!*}}\xspace}
\newcommand{\CCTLstar}{\texttt{CCTL\textsuperscript{\!*}}\xspace}
\newcommand{\CLTL}{\texttt{CLTL}\xspace}

\newcommand{\RAE}{\texttt{RAE}\xspace}

\newcommand{\U}{\ensuremath{\operatorname{\mathtt{U}}}\xspace}

\newcommand{\X}{\ensuremath{\operatorname{\mathtt{X}}}\xspace}
\newcommand{\G}{\ensuremath{\operatorname{\mathtt{G}}}\xspace}
\newcommand{\F}{\ensuremath{\operatorname{\mathtt{F}}}\xspace}
\newcommand{\E}{\ensuremath{\operatorname{\mathtt{E}}}\xspace}
\newcommand{\A}{\ensuremath{\operatorname{\mathtt{A}}}\xspace}

\newcommand{\true}{\top}
\newcommand{\false}{\bot}

\newcommand{\defeq}{\overset{\textsf{def}}{\Leftrightarrow}}
\newcommand{\define}[1]{\overset{\textsf{def}}{#1}}

\newcommand{\MC}[1]{\texttt{MC}(#1)}

\newcommand{\PresHartig}{\mathtt{PH}}
\newcommand{\PH}{\texttt{PH}\xspace}

\newcommand{\Guards}[1]{\mathfrak{G}(#1)}
\newcommand{\sub}{\mathtt{sub}}

\renewcommand{\path}{P}

\renewcommand{\st}{\mathtt{st}}
\newcommand{\labels}{\mathtt{lab}}
\newcommand{\guards}{\mathtt{g}}
\newcommand{\update}{\mathtt{u}}
\newcommand{\type}{\mathtt{t}}

\newcommand{\bal}{\mathtt{bal}}
\newcommand{\loc}{\mathtt{loc}}
\newcommand{\comp}{\mathtt{comp}}
\newcommand{\succpos}{\mathtt{succ}}
\newcommand{\row}{\mathtt{row}}
\newcommand{\chk}{\mathtt{chk}}

\tikzset{
  >=stealth',
  initial text =,
  node distance = 5em,
  tiny nodes/.style = {
     every state/.style = {
       inner sep = 0pt,
       minimum size = 1ex,
       initial distance = 2ex,
       node distance = 5.5ex
     },
    font=\small,
    every edge/.style = {draw, shorten >=1pt, font=\tiny},
    every label/.style = {font=\tiny}
  },
  small nodes/.style = {
     every state/.style = {
       inner sep = 0pt,
       minimum size = 3.5ex,
       initial distance = 3ex,
       node distance = 3.8em
     },
    font=\small,
    every edge/.style = {
      draw,
      shorten >=1pt,
      font=\tiny,
    },
  },
  loop above left/.style = {
    out=75, in=105, loop
  }
}

\declaretheorem[sibling=theorem, style=plain]{Lemma}
\declaretheorem[sibling=theorem, style=plain]{Theorem}
\declaretheorem[sibling=theorem, style=plain]{Definition}
\declaretheorem[sibling=theorem, style=plain]{Corollary}
\declaretheorem[sibling=theorem, style=plain]{Example}

\makeatletter
\renewcommand\subparagraph{\@startsection{subparagraph}{5}{\z@}                                       {1.25ex \@plus1ex \@minus .2ex}                                       {-1em}                                      {\sffamily\normalsize\bfseries}}
\makeatother

\begin{document}

\maketitle

\begin{abstract}
We study several extensions of linear-time and computation-tree
temporal logics with quantifiers that allow for counting how often certain
properties hold.
For most of these extensions, the model-checking problem is undecidable, but we show that decidability can be recovered by considering flat Kripke structures where each state belongs to at most one simple loop.
Most decision procedures are based on results on (flat) counter systems where counters are used to implement the evaluation of counting operators.
\end{abstract}

\section{Introduction}

Model checking \cite{DBLP:journals/cacm/ClarkeES09} is a method to verify automatically the correct
behaviour of systems.
It takes as input a model of the system to be verified and a logical formula encoding the specification and checks whether the behaviour of the model satisfies the
formula. One key aspect of this method is to find the appropriate balance
between expressiveness of models and logical formalisms and
efficiency of the model-checking algorithms.  If the model is too
expressive, e.g.\ Turing machines, then the model-checking problem,
even with very simple logical formalisms, becomes undecidable.  On the
other hand, some expressive logics have been proposed in order to
reason on the temporal executions of simple models such as Kripke
structures.
This is the case for the linear temporal logic \LTL \cite{DBLP:conf/focs/Pnueli77}
and the branching-time temporal logics \CTL~\cite{DBLP:conf/lop/ClarkeE81} and \CTLstar~\cite{DBLP:conf/popl/EmersonH83}, for which the model-checking problem
has been shown to be \textsc{PSpace}-complete, contained in \textsc{P} and \textsc{PSpace}-complete, respectively (see, e.g., \cite{DBLP:books/daglib/0020348}).

 Even though  these  logical formalisms
allow for stating classical properties like safety or liveness over
executions of Kripke structures, their expressiveness is limited. In particular they cannot describe quantitative
aspects, as for instance the fact that a property has been
true twice as often as another along an execution.
One approach to solve this issue is to
extend the logic with some ability to \emph{count} positions of an execution satisfying some property
and to check constraints over such numbers at some positions.
Such a counting extension is proposed in
 \cite{DBLP:journals/corr/abs-1211-4651} for \CTL leading to a logic
   denoted here as \CC.
 This formalism can state properties such as an event $p$ will eventually occur and before that, the number of events $q$ is larger than two.
 The authors propose further an extension called (here) $\CCpm$ that admits diagonal comparisons (i.e., negative and positive coefficients) to state, for instance that the number of events $b$ is greater than the number of
 events $c$.
 It is shown that the model-checking problem for $\CC$ is decidable in polynomial time and that the satisfiability problem for $\CCpm$
 is undecidable.
 A similar extension for \LTL is considered in \cite{DBLP:conf/time/LaroussinieMP10} where it is proven that model checking of \CL is \textsc{ExpSpace}-complete while that of $\CLpm$
 is undecidable.

Following the same motivation, \emph{regular availability expressions (\RAE)} were introduced in \cite{DBLP:conf/concur/HoenickeMO10} extending regular expressions by a mechanism to express that on a (sub-)word matching an expression specific letters occur with a given relative frequency.
Unfortunately, emptiness of the intersection of two such expressions was shown undecidable.
Even for single expressions only a non-elementary procedure is known for verification (inclusion in regular languages) and deciding emptiness \cite{DBLP:conf/fsttcs/AbdullaAMS15}.
The case is similar for the logic \fLTL \cite{DBLP:conf/tase/BolligDL12}, a variant of \LTL that features an until operator extended by a frequency constraint.
The operator is intended to relax the classical semantics where $\varphi \U \psi$ requires $\varphi$ to hold at all positions before $\psi$.
For example, the \fLTL formula $p\U^\frac{1}{3} q$ states that $q$ holds eventually and before that the proportion of positions satisfying $p$ should be at least one third.
The concept of relative frequencies embeds naturally into the context of counting logics as it can be understood as a restricted form of counting.
In fact, \fLTL can be considered as a fragment of $\CLpm$ and still has an undecidable satisfiability problem~\cite{DBLP:conf/tase/BolligDL12} implying the same for model-checking Kripke structures.
Moreover, most techniques employed for obtaining results on \RAE as well as \fLTL involve variants of counter systems.

Looking at the model-checking problem from the model point of view, recent work has shown that restrictions can be imposed on Kripke structures to obtain better
complexity bounds.
As a matter of fact if the structure is \emph{flat} (or weak), which means every state belongs to at most one simple cycle in the graph underlying the structure, then the model-checking problem for \LTL becomes \textsc{NP}-complete \cite{DBLP:conf/concur/KuhtzF11}.
Such a restriction has as well been successfully applied to more complex classes of models.
It is well known that the reachability problem for two-counter systems is undecidable \cite{minsky-computation-67} whereas for flat systems the problem is decidable for any number of counters~\cite{DBLP:conf/fsttcs/FinkelL02}, even more, model checking of \LTL is \textsc{NP}-complete~\cite{DBLP:journals/iandc/DemriDS15}.
Flat structures are not only interesting because of their algorithmic properties, but also because they can be used as a way to under-approximate the behaviour of non-flat systems.
For instance for counter systems one gets a semi-decision procedure for the reachability problem which consists in enumerating flat sub-systems and testing for reachability.
In simple words, flat structures can be understood as an extension of paths typically used in bounded model checking and we expect that bounded model checking using flat structures rather than paths improves practical model checking approaches.

\subparagraph{Contributions.}

We consider the model-checking problem for a counting logic that we call $\CCTLstar$ where we use variables to mark positions on a run from where we begin to count the number of times a subformula is satisfied.
Such a way of counting was also introduced in~\cite{DBLP:journals/corr/abs-1211-4651}, see \cref{ssc:counting-logics} for a comparison.
We study as well its fragments \fCTL, \fLTL and \fCTLstar where the explicit counting mechanism is replaced by a generalized version of the until operator capable of expressing frequency constraints.

First we prove that \fCTL model checking is at most exponential in the formula size and polynomial in the structure size by using an algorithm similar to the one for \CTL model checking.
To deal with frequency constraints a counter is employed for tracking the number of times a subformula is satisfied in a run of a Kripke structure.
We then show that for flat Kripke structures the model-checking problems of $\fLTL$ and $\CCTLstar$ are decidable.
For the former, our method is a guess and check procedure based on the existence of a flat counter system as witness of a run of the Kripke structure satisfying the \fLTL formula.
For the latter, we use a technique which consists in encoding the run of a flat Kripke
structure into a Presburger arithmetic formula and then we show that model checking of $\CCTLstar$ can be translated into the satisfiability problem of a decidable extension of Presburger arithmetic, called \PH, featuring a counting quantifier known as Härtig quantifier.
We hence provide new decidability results for $\CCTLstar$ which in practice could be used as an under-approximation approach to the general model-checking problem.
We furthermore relate an extension of Presburger arithmetic, for which the complexity of the
satisfiability problem is open, to a concrete model-checking problem.
In summary, for model checking different fragments of $\CCTLstar$ on Kripke structures (\KS) or flat Kripke structures (\FKS) we obtain the picture shown in \cref{tab:results} where bold entries are our novel results.
\begin{table}[tb]
\begin{center}
\begin{tabular}{@{} c @{\hspace{2ex}} c @{\hspace{2ex}} c @{\hspace{2ex}} c @{\hspace{2ex}} c @{\hspace{2ex}} c @{\hspace{2ex}} c @{\hspace{2ex}} c @{\hspace{2ex}} c @{\hspace{2ex}} c @{}}
\toprule &
\CTL & \LTL &\CTLstar & $\fLTL$ & $\fCTL$ & $\fCTLstar$     & $\CLTL$ & $\CCTL$ & $\CCTLstar$ \\

\midrule $\KS$ &
\textsc{P} & \textsc{PSpace}-c. & \textsc{PSpace}-c. & undec. \cite{DBLP:conf/tase/BolligDL12} & \textbf{\textsc{Exp}} & undec. & undec. & undec.  \cite{DBLP:journals/corr/abs-1211-4651} & undec.  \\

$\FKS$ &
\textsc{P} & \textsc{NP}-c. \cite{DBLP:conf/concur/KuhtzF11} & \textsc{PSpace} &  \textbf{\textsc{NExp}} & \textbf{\textsc{Exp}} & \textbf{\textsc{ExpSpace}} & \textbf{\textsc{PH}}  & \textbf{\textsc{PH}} & \textbf{\textsc{PH}} \\

\bottomrule
\end{tabular}
\caption{
  Complexity characterisation of the model-checking problems of fragments of $\CCTLstar$. \textsc{PH} indicates polynomial reducibility to the (decidable) satisfiability problem of \PH.
      \label{tab:results}
}
\end{center}\vspace{-5ex}
\end{table}

\section{Definitions}

\subsection{Preliminaries}

We write $\mathbb{N}$ and $\mathbb{Z}$ to denote the sets of natural numbers (including zero) and integers, respectively, and $\interval{i}{j}$ for $\set{k \in \mathbb{Z} \mid i \le k \le j}$.
We consider integers encoded with a binary representation.
For a finite alphabet $\Sigma$, $\Sigma^*$ represents the set of finite words over $\Sigma$,
$\Sigma^+$ the set of finite non-empty words over $\Sigma$ and $\Sigma^\omega$ the set of
infinite words over $\Sigma$.
  For a finite set $E$ of elements, $\card{E}$ represents its cardinality.
For (finite or infinite) words and general sequences $u=a_0a_1…a_k…$ of length at least $k+1>0$ we denote by $u(k)=a_k$ the $(k+1)$-th element and refer to its indices $0,1,…$ as positions on $u$.
If $u$ is finite then $\length{u}$ denotes its length.
For arbitrary functions $f: A \to B$ and elements $a\in A,b\in B$ we denote by $f[a\mapsto b]$ the function $f'$ that is equal to $f$ except that $f'(a)=b$.
We write $\mathbf{0}$ and $\mathbf{1}$ for the functions $f_0:A\to\lbrace0\rbrace$ and $f_1:A\to\lbrace1\rbrace$, respectively, if the domain $A$ is understood. By $B^A$ for sets $A$ and $B$ we denote the set of all functions from $A$ to $B$.

\subparagraph{Kripke structures.}
Let $AP$ be a finite set of \emph{atomic propositions}.
A \emph{Kripke structure} is a tuple $\mathcal{K}=(S,s_I, E, \lambda)$ where $S$ is a finite set of control states, $s_I\in S$ the initial control state, $E \subseteq S \times S$ the set of edges and $\lambda: S \mapsto 2^{AP}$ the labelling function.
A finite \emph{path} in $\mathcal{K}$ is a sequence $u=\state_0 \state_1 \ldots \state_k \in S^+$ with $(\state_i,\state_{i+1}) \in \Edges$ for all $i \in \interval{0}{k-1}$.
Infinite paths are defined analogously.
A \emph{run} $\run$ of $\mathcal{K}$ is an infinite path with $\run(0)=s_I$.
We denote by $\Runs(\mathcal{K})$ the set of runs of $\mathcal{K}$.
Due to the single initial state, we assume without loss of generality that the graph of $\mathcal{K}$ is connected, i.e.\ all states are reachable.
A \emph{simple loop} in $\mathcal{K}$ is a finite path $u=\state_0 \state_1 \ldots \state_k$ such that $i\neq j$ implies $\state_i \neq \state_j$ for all $i,j \in \interval{0}{k}$ and $(\state_k,\state_0)\in E$.
A Kripke structure $\mathcal{K}$ is called \emph{flat} if for each state $\state \in \States$ there is at most one simple loop $u$ in $\mathcal{K}$ with $u(0)=\state$.
See \cref{fig:fks} for an example.
The classes of all Kripke structures and all flat Kripke structures are denoted \KS and \FKS, respectively.

\subparagraph{Counter systems.}

Our proofs use systems with integer counters and simple guards.
A \emph{counter system} is a tuple $\mathcal{S}=(S,s_I,C,\Delta)$ where $S$ is a finite set of control states, $s_I\in S$ is the initial state, $C$ is a finite set of counter names and $\Delta \subseteq S \times \mathbb{Z}^C \times 2^{\Guards{C}} \times S$ is the transition relation where $\Guards{C}=\lbrace(c<0),(c\ge0)\mid c\in C\rbrace$.
An infinite sequence $s_0s_1…\in S^\omega$ of states starting in $s_0=s_I$ is called a \emph{run} of $\mathcal{S}$ if there is a sequence $\theta_0\theta_1…\in(\mathbb{Z}^C)^\omega$ of valuation functions $\theta_i:C\to\mathbb{Z}$ with $\theta_0=\mathbf{0}$ and a transition $(s_i, \vec{u}_i, G_i, s_{i+1})\in\Delta$ for every $i\in\mathbb{N}$ such that $\theta_{i+1} = \theta_i+\vec{u}_i$
(defined point-wise as usual), $\theta_{i+1}(c)<0$ if $(c<0)\in G_i$ and $\theta_{i+1}(c)\ge0$ if $(c\ge0)\in G_i$ for all $c\in C$.
Again, we denote by $\Runs(\mathcal{S})$ the set of all such runs and assume the graph of control states underlying $\mathcal{S}$ is connected.

\subsection{Temporal Logics with Counting}
\label{ssc:counting-logics}

We now introduce the different formalisms we use in this work as specification language. The most general one is the branching-time logic $ \CCTLstar $
which extends the branching-time logic $\CTLstar $ (see e.g.~\cite{DBLP:books/daglib/0020348})
with the following features: it has operators that allow for counting
along a run the number of times a formula is satisfied
and which stores the result into a variable.
The counting starts when the associated variable is “placed” on the run.
These variables may be shadowed by nested quantification, similar to the semantics of the freeze quantifier in linear temporal logic \cite{DBLP:journals/tocl/DemriL09}.

Let $V$ be a set of \emph{variables} and $AP$ a set of atomic propositions.
The syntax of \CCTLstar formulae $\varphi$ over $V$ and $AP$ is given by the grammar rules
\[
  \varphi ::= p ~\mid~ \varphi \land \varphi ~\mid~ \neg\varphi  ~\mid~ \X \varphi ~\mid~ \varphi \U \varphi ~\mid~ \E \varphi  ~\mid~ x.\varphi ~\mid~ \tau\le\tau
  \qquad
  \tau ::= a\mid a\cdot\#_x(\varphi)\mid\tau+\tau
  \]
for $p \in AP$, $x\in V$ and $a\in \mathbb{Z}$.
Common abbreviations such as $\true \equiv p\lor\neg p$, $\false \equiv \neg\true$, $\F\varphi \equiv \true \U \varphi$, $\G\varphi \equiv \neg\F\neg\varphi$ and $\A \varphi \equiv \neg\E\neg\varphi$ may also be used.
The set of all \emph{subformulae} of a formula $\varphi$ (including itself) is denoted $\sub(\varphi)$ and $|\varphi|$ denotes the length of $\varphi$, with binary encoding of numbers.

\subparagraph{Semantics.}

Intuitively, a variable $x$ is used to mark some position on the concerned run.
Within the scope of $x$ a term $\#_{x}(\varphi)$ refers to the number of times the formula $\varphi$ holds between the current position and that marked by $x$.
The semantics of \CCTLstar is hence defined with respect to a Kripke structure $\mathcal{K}=(S,s_I,E,\lambda)$, a run $\rho\in\Runs(\mathcal{K})$, a position $i\in\mathbb{N}$ on $\rho$ and a valuation function $\theta: V \to \mathbb{N}$ assigning a position (index) on $\rho$ to each variable.
The satisfaction relation $\models$ is defined inductively for $p\in AP$, formulae $\varphi,\psi$ and terms $\tau_1,\tau_2$ by
\[\begin{array}{lcl}
 (\rho,i,\theta) \models p     & \defeq & p\in\lambda(\rho(i)),     \\
 (\rho,i,\theta) \models \X \varphi  & \defeq & (\rho,i+1,\theta) \models \varphi,  \\
 (\rho,i,\theta) \models \varphi \U \psi  & \defeq & \exists k\ge i:(\rho,k,\theta) \models \psi \text{ and } \forall j\in[i,k-1]: (\rho,j,\theta) \models\varphi, \\
(\rho,i,\theta) \models \E \varphi & \defeq & \exists\rho'\in\Runs(\mathcal{K}):\forall j\in[0,i]:\rho'(j)=\rho(j) \text{ and } (\rho',i,\theta)\models\varphi, \\
(\rho,i,\theta) \models x.\varphi & \defeq & (\rho,i,\theta[x\mapsto i]) \models \varphi, \\
(\rho,i,\theta) \models \tau_1\le\tau_2 & \defeq & \llbracket\tau_1\rrbracket(\rho,i,\theta) \le \llbracket\tau_2\rrbracket(\rho,i,\theta),
\end{array}\]
where the Boolean cases are omitted and the semantics of terms is given, for $a\in\mathbb{Z}$, by
 \[\begin{array}{rcl}
   \llbracket a\rrbracket(\rho,i,\theta) &\define{=}& a,\\
   \llbracket\tau_1+\tau_2\rrbracket(\rho,i,\theta) & \define{=}& \llbracket\tau_1\rrbracket(\rho,i,\theta)+\llbracket\tau_2\rrbracket(\rho,i,\theta),\\
   \llbracket a \cdot \#_x(\varphi)\rrbracket(\rho,i,\theta) & \define{=}& a\cdot|\lbrace j\in\mathbb{N}\mid\theta(x)\le j \le i, (\rho,j,\theta)\models\varphi\rbrace|.
\end{array}\]
We abbreviate $(\rho,i,\mathbf{0}) \models\varphi$ by $(\rho,i) \models \varphi$ and $(\rho,0)\models\varphi$ by $\rho\models\varphi$ and say that $\rho$ satisfies $\varphi$ (at position $i$) in these cases.
Moreover, we say a state $s\in S$ satisfies $\varphi$, denoted $s\models\varphi$ if there are $\rho_s\in\Runs(\mathcal{K})$ and $i\in\mathbb{N}$ such that $\rho_s(i)=s$ and $(\rho_s,i)\models\varphi$.
The Kripke structure $\mathcal{K}$ satisfies $\varphi$, denoted by $\mathcal{K} \models \varphi $, if $s_I \models \varphi$.
Note that we choose to define the model-checking relation existentially but since the formalism is closed under negation, this does not have major
consequences on our results.

\subparagraph{Fragments.}
We define the following fragments of $\CCTLstar$ in analogy to the classical logics $\LTL$ and $\CTL$.
The \emph{linear} time fragment $\CLTL$ consists of those \CCTLstar formulae that do not use the path quantifiers $\E$ and $\A$.
The \emph{branching} time logic \CCTL restricts the use of temporal operators $\X$ and $\U$  such that each occurrence must be preceded immediately by either $\E$ or $\A$.
Similar branching-time logics have been considered in \cite{DBLP:journals/corr/abs-1211-4651}.

\subparagraph{Frequency logics.}
A major subject of our investigation are frequency constraints.
This concept embeds naturally into the context of counting logics as it can be understood as a restricted form of counting.
We therefore define in the following the frequency temporal logics \fCTLstar, \fLTL and \fCTL as fragments of \CCTLstar.
Consider the following grammar defining the syntax of formulae $\varphi$ for natural numbers $n,m\in\mathbb{N}$ with $n\le m>0$ and $p\in AP$.
\begin{align*}
  \varphi ::= p \mid \varphi \land \varphi \mid \neg\varphi \mid \alpha &&
  \beta ::= \X\varphi  \mid \varphi\U^{\frac{n}{m}} \varphi \end{align*}
With the additional rule $\alpha::= \E \varphi\mid\beta$ it defines precisely the set of \fCTLstar formulae while it defines \fCTL for $\alpha ::= \E \beta\mid\A \beta$ and \fLTL for $\alpha ::= \beta$.
The semantics is defined by interpreting \fCTLstar formulae as \CCTLstar with the additional equivalence
\begin{equation}\label{eqn:funtil-semantics}
  \varphi\U^\frac{n}{m}\psi \define{\equiv} \psi \lor x\,.\,\F\ ((\X \psi) \, \land \, m\cdot\#_x(\varphi) \ge n\cdot\#_x(\top))
\end{equation}
for \fCTLstar formulae $\varphi$ and $\psi$ and a variable $x\in V$ not being used in either $\varphi$ or $\psi$.

\begin{figure}[t]
\begin{center}
  
\begin{tikzpicture}[small nodes]
    \node[state, initial, label=below:$p$] (q0) {$s_0$};
    \node[state, right of = q0] (q1) {$s_1$};
    \node[state, right of = q1] (q2) {$s_2$};
    \node[state, above of = q2, node distance = 3em, label=below:$r$] (q3) {$s_3$};
    \node[state, right of = q2, label=below:$r$] (q4) {$s_4$};
    \node[state, right of = q4, label=below:$q$] (q5) {$s_5$};

    \path[->] (q0) edge (q1)
                  (q1) edge (q2)
                  (q2) edge (q4)
                  (q2) edge[bend right=40] (q3)
                  (q3) edge[bend right=40] (q2)
                  (q0) edge[loop above left] (q0)
                  (q0) edge[bend right=50] (q2)
                  (q4) edge (q5)
                  (q4) edge[loop above left] (q4)
                  (q5) edge[loop above left] (q5);

\end{tikzpicture}
 \end{center}
    \caption{A flat Kripke over $AP=\lbrace p,q,r\rbrace$.\label{fig:fks}}
\end{figure}

\begin{Example}
  Consider  the Kripke structure given by \cref{fig:fks} and the \CCTL formula $\varphi_1=z.\A\G \left(q \to ( \#_z(p) \le \#_z(\E\X r))\right)$.
It basically states that on every path reaching $s_5$ there must be a position where the states $s_2$ and $s_4$ (satisfying $\E\X r$) together have been visited at least as often as the state $s_0$.
A different, yet similar statement can be formulated using only frequency constraints:
$\varphi_1' = \A ((\E\X r )\U^{\frac{1}{2}} q)$ states that $s_5$ must always be reached while visiting $s_2$ and $s_4$ together at least as often as $s_0$, $s_1$ and $s_3$.
Both $\varphi_1$ and $\varphi_1'$ are violated, e.g.\ by the path $s_0^3s_1s_2s_4s_5^\omega$.
The Kripke structure however satisfies $\varphi_2=z.\A\G \left(\neg q \to \E\F \#_z(p) < \#_z(r)\right)$ because from every state except $\mathrm{s_5}$ the number of positions that satisfy $r$ can be increased arbitrary without increasing the number of those satisfying $p$.
Notice that this would not be the case, e.g., if $\mathrm{s_4}$ was labelled by $p$.
\end{Example}

While the positional variables in \CCTLstar are a very flexible way of defining the scope of a constraint, frequency constraints in \fCTLstar are always bound to the scope of an until operator.
The same applies to the counting constraints of \CL as defined in \cite{DBLP:journals/corr/abs-1211-4651}.
For example, the \CL formula $\varphi\U_{[a_1\#(\varphi_1) + \cdots +a_n\#(\varphi_n) \ge k]}\psi$ is equivalent to the \CLTL formula $z.\varphi\U(\psi\land a_1\#_z(\varphi_1) + \cdots +a_n\#_z(\varphi_n) \ge k)$.
Admitting only natural coefficients, \CL can be encoded even in \LTL making it thus strictly less expressive than \fLTL.
On the other hand, \CLpm admits arbitrary integer coefficients, which is more general than the frequency until operator of \fLTL.
For example, $p\U^{\frac{a}{b}}q$ can be expressed as $\top\U_{[b\#(p) -a\#(\top) \ge 0]}q$ in {\CLpm}. The relation between \CCpm and \fCTL, as well as \CCstarpm and \fCTLstar is analogous.

\subparagraph{Model-checking problem.}

We now present the problem on which we focus our attention.
The \emph{model-checking problem} for a class $\mathfrak{K}\subseteq\KS$ of Kripke structures and a specification language $\mathcal{L}$ (in our case all the specification languages are fragments of $\CCTLstar$) is denoted by $\MC{\mathfrak{K},\mathcal{L}}$ and defined as the following decision problem.

\vspace{-0.5ex}

\begin{itemize}
  \item[\lipicsbf{Input:}] A Kripke structure $\mathcal{K} \in \mathfrak{K}$ and a formula $\varphi \in \mathcal{L}$. \qquad
    \lipicsbf{Decide:} Does $\mathcal{K} \models \varphi $ hold?
\end{itemize}

\vspace{-0.5ex}

For temporal logics without counting variables, the model-checking problem over Kripke structure has been studied intensively and is known to be \textsc{PSpace}-complete for \LTL and \CTLstar and in \textsc{P} for \CTL (see e.g.~\cite{DBLP:books/daglib/0020348}).
It has recently been shown that when restricting to flat (or weak) structures the complexity of the model-checking problem for \LTL is lower than in the general case \cite{DBLP:conf/concur/KuhtzF11}:
it drops from \textsc{PSpace} to \textsc{NP}.
As we show later, in the case of $\CCTLstar$, flatness of the structures allows us to regain decidability of the model-checking problem which is in general undecidable.
In this paper, we propose various ways to solve the model-checking problem of fragments of \CCTLstar over flat structures.
For some of them we provide a direct algorithm, for others we reduce our problem to the satisfiability problem of a decidable extension of Presburger arithmetic.

\section{Model-checking Frequency \CTL}
\label{sec:fctl}

Satisfiability of \fLTL is undecidable \cite{DBLP:conf/tase/BolligDL12} implying the same for model-checking \fLTL, \CLTL and \CCTLstar over Kripke structures.
This applies moreover to \CCTL \cite{DBLP:journals/corr/abs-1211-4651}.
In contrast, we show in the following that \MC{\KS, \fCTL} is decidable using an extension of the well-known labelling algorithm for \CTL (see e.g.\  \cite{DBLP:books/daglib/0020348}).

Let $\mathcal{K}=(S,s_I,E,\lambda)$ be a Kripke structure and $\Phi$ an \fCTL formula.
We compute recursively subsets $S_\varphi\subseteq S$ of the states of $\mathcal{K}$ for every subformula $\varphi\in\sub(\Phi)$ of $\Phi$ such that for all $s\in S$ we have $s\in S_\varphi$ iff $s\models\varphi$.
Checking whether the initial state $s_I$ is contained in $S_\Phi$ then solves the problem.
Propositions ($p\in AP$), negation ($\neg\varphi$), conjunction ($\varphi\land\psi$) and \emph{temporal next} ($\E \X \varphi$, $\A\X\varphi$) are handled as usual, e.g.\ $S_p =\lbrace q\in S\mid p\in\lambda(q)\rbrace$ and $S_{\E\X \varphi} =\lbrace q\in S\mid\exists q'\in S_\varphi: (q,q')\in\delta\rbrace$.

To compute if a state $s\in S$ satisfies a formula of the form $\E\varphi\U^r\psi$ or $\A\varphi\U^r\psi$, assume that $S_\varphi$ and $S_\psi$ are given inductively.
If $s\in S_\psi$ we immediately have $s\in S_{\E\varphi\U^r\psi}$ and $s\in S_{\A\varphi\U^r\psi}$.
For the remaining cases, the problem of deciding whether $s\in S_{\E\varphi\U^r\psi}$ or $s\in S_{\A\varphi\U^r\psi}$, respectively, can be reduced in linear time to the \emph{repeated control-state reachability} problem in systems with one integer counter.
The idea is to count the ratio along paths $\rho\in S^\omega$ in $\mathcal{K}$ as follows, in direct analogy to the semantics defined in \cref{eqn:funtil-semantics}.
Assume $r=\frac{n}{m}$ for $n,m\in\mathbb{N}$ and $n\le m$.
For passing any position on $\rho$ we pay a fee of $n$ and for those positions that satisfy $\varphi$ we gain a reward of $m$.
Thus, we obtain a non-negative balance of rewards and gains at some position on $\rho$ if, in average, among every $m$ positions there are at least $n$ positions that satisfy $\varphi$, meaning the ratio constraint is satisfied.
In $\mathcal{K}$, this balance along a path can be tracked using an \emph{integer counter} that is increased by $m-n$ when leaving a state $s'\in S_\varphi$ and decreased by adding $-n$ whenever leaving a state $s'\not\in S_\varphi$.
Thus, let $\hat{\mathcal{K}}_s=(S,s,\lbrace c\rbrace,\Delta)$ be the counter system with
\[
  \Delta = \lbrace(t,\mathbf{u},\emptyset,t')\mid(t,t')\in E,\ t\not\in S_\varphi \Rightarrow \mathbf{u}(c)=-n,\ t\in S_\varphi \Rightarrow \mathbf{u}(c)=m-n\rbrace.
\]

The state $s$ satisfies the formula $\A\varphi\U^r\psi$ if there is no path starting in state $s$ violating the formula  $\varphi\U^r\psi$.
The latter is the case if at every position where $\psi$ holds, the balance computed up to this position is negative.
Therefore, consider an extension $\mathcal{R}_s$ of $\hat{\mathcal{K}}_s$ where every edge leading into a state $s'\in S_\psi$ is guarded by the constraint $c<0$.
Every (infinite) run of $\mathcal{R}_s$ is now a counter example for the property holding at $s$. To decide whether $s\in S_{\A\varphi\U^r\psi}$ it suffices to check that in $\mathcal{R}_s$ no state is repeatedly reachable from $s$.

A formula $\E\varphi\U^r\psi$ is satisfied by $s$ if there is some state $s'\in S_\psi$ reachable from $s$ with a non-negative balance.
Hence, consider the counter system $\mathcal{U}_s=(S\uplus\lbrace\mathtt{t}\rbrace, s,\lbrace c\rbrace, \Delta')$ obtained from $\hat{\mathcal{K}}_s$ featuring a new sink state $\mathtt{t}\not\in S$.
The transition relation
\[
    \Delta'=\Delta \cup\lbrace(s',\mathbf{0},\lbrace c\ge0\rbrace,\mathtt{t})\mid s'\in S_\psi\rbrace\cup\lbrace(\mathtt{t},\mathbf{0},\emptyset,\mathtt{t})\rbrace
\]
extends $\Delta$ such that precisely the paths starting in $s$ and reaching a state $s'\in S_\psi$ with non-negative counter value (i.e.\ sufficient ratio) can be extended to reach $\mathtt{t}$.
Checking if $s$ is supposed to be contained in $S_{\E\varphi\U^r\psi}$ then amounts to decide whether $\mathtt{t}$ is (repeatedly) reachable from $s$ in $\mathcal{U}_s$.

Finally, repeated reachability is easily translated to the \emph{accepting run} problem of \emph{Büchi pushdown systems (BPDS)} and the latter is in \textsc{P} \cite{DBLP:conf/concur/BouajjaniEM97}.
A counter value $n\ge0$ can be encoded into a stack of the form $\oplus^n$ while $\ominus^n$ encodes $-n\le0$ and for evaluating the guards $c\ge0$ and $c<0$ only the top symbol is relevant.
Simulating an update of the counter by a number $a\in\mathbb{Z}$ requires to perform $|a|$ push or pop actions.
The size of the system is therefore linear in the largest absolute update value and hence exponential in its binary representation.
Since the updates of the constructed counter systems originate from the ratios in $\Phi$, the corresponding BPDS are of up to exponential size in $|\Phi|$.
During the labelling procedure this step must be performed at most a polynomial number of times giving an exponential-time algorithm.

\begin{restatable}{Theorem}{mcfctl}
\label{thm:mc-fctl}
    \MC{\KS, \fCTL} is in \textsc{Exp}.
\end{restatable}

It is worth noting that for a fixed formula (program complexity) or a unary encoding of numbers in frequency constraints, the size of the constructed Büchi pushdown systems and thus the runtime of the algorithm remains polynomial.

\begin{Corollary}
    \MC{\KS, \fCTL} with unary number encoding is in \textsc{P}.
\end{Corollary}

\section{Model-checking Frequency \LTL over Flat Kripke Structures}

We show in this section that model-checking \fLTL is decidable over flat Kripke structures.
As decision procedure we employ a \emph{guess and check} approach:
given a flat Kripke structure $\mathcal{K}$ and an \fLTL formula $\Phi$, we choose non-deterministically a set of satisfying runs to witness $\mathcal{K}\models\Phi$.
As representation for such sets we introduce \emph{augmented path schemas} that extend the concept of path schemas \cite{DBLP:conf/atva/LerouxS05,DBLP:journals/iandc/DemriDS15} and provide for each of its runs a labelling by formulae.
We show that if an augmented path schema features a syntactic property that we call \emph{consistency} then the associated runs actually satisfy the formulae they are labelled with.
Moreover, we show that every run of $\mathcal{K}$ is in fact represented by some consistent schema of size at most exponential in $|\mathcal{K}|+|\Phi|$.
This gives rise to the following non-deterministic procedure.
\begin{enumerate}
\item \lipicsbf{Read as input} an FKS $\mathcal{K}$ and an \fLTL formula $\Phi$.
\item \lipicsbf{Guess} an augmented path schema $\mathcal{P}$ in $\mathcal{K}$ of at most exponential size.
\item \lipicsbf{Terminate} successfully if $\mathcal{P}$ is consistent and accepts a run that is initially labelled by $\Phi$.
\end{enumerate}

We fix for this section a flat Kripke structure $\mathcal{K}=(S,s_I,E, \lambda)$ and an \fLTL formula $\Phi$.
For convenience we assume that $AP\subseteq\sub(\Phi)$.
Omitted technical details can be found in \cref{apx:correctness,apx:completeness}.
 
\subsection{Augmented Path Schemas}
\label{ssc:aps}

The set of runs of $\mathcal{K}$ can be represented as a finite number of so-called path schemas that consist of a sequence of paths and simple loops consecutive in $\mathcal{K}$~\cite{DBLP:conf/atva/LerouxS05,DBLP:journals/iandc/DemriDS15}.
A path schema represents all runs that follow the given shape while repeating each loop arbitrarily often.
For our purposes we extend this idea with additional labellings and introduce integer counters, updates and guards that can restrict the admitted runs.

\begin{definition}[Augmented Path Schema]
An \emph{augmented state} of $\mathcal{K}$ is a tuple $a = (s, L, G, \vec{u}, t)\in S\times2^{\sub(\Phi)}\times2^{\Guards{C}}\times\mathbb{Z}^C\times\lbrace\mathtt{L},\mathtt{R}\rbrace$ comprised of a state $s$ of $\mathcal{K}$, a set of formula labels $L$, guards $G$ and an update $\mathbf{u}$ over a set of counter names $C$, and a \emph{type} indicating whether the state is part of a loop ($\mathtt{L}$) or a not ($\mathtt{R}$).
We denote by $\st(a)=s$, $\labels(a)=L$, $\guards(a)=G$, $\update(a)=\vec{u}$ and $\type(a)=t$ the respective components of $a$.
An \emph{augmented path} in $\mathcal{K}$ is a sequence $u=a_0…a_n$ of augmented states $a_i$
 such that $(\st(a_i),\st(a_{i+1}))\in E$ for $i\in[0,n-1]$.
If $\type(a_i)=\mathtt{R}$ for all $i\in[0,n-1]$ then $u$ is called a \emph{row}.
It is called an \emph{augmented simple loop} (or simply \emph{loop})  if it is non-empty and $(\st(a_n),\st(a_1))\in E$ and $\st(a_i)\neq \st(a_j)$ for $i\neq j$ and  $\type(a_i)=\mathtt{L}$ for all $i\in[0,n-1]$.

An \emph{augmented path schema (APS)} in $\mathcal{K}$ is a tuple $\mathcal{P}=(P_0,…,P_n)$ where each component $P_k$ is a row or a loop, $P_n$ is a loop and their concatenation $P_1P_2…P_n$ is an augmented path.
\end{definition}

Thanks to counters we can, for example, restrict to those runs satisfying a specific frequency constraint at some positions tracking it as discussed in \cref{sec:fctl}.
\Cref{fig:example-ps-1} shows an example of an APS with edges indicating the possible state progressions.
It features a single counter that tracks the frequency constraint of a formula $r\U^{\frac{2}{3}}q$ from state $1$.

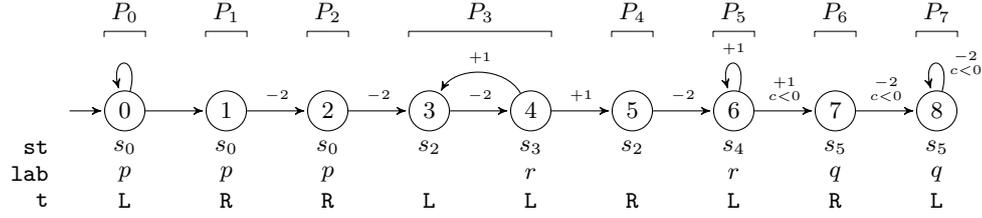
\begin{figure}[tb]
\begin{center}

\begin{tikzpicture}[small nodes]
    \node[state, initial] (p1) {$0$};
            \node[state, right of = p1] (p2) {$1$};
    \node[state, right of = p2] (p3) {$2$};
    \node[state, right of = p3] (p4) {$3$};
    \node[state, right of = p4] (p5) {$4$};
    \node[state, right of = p5] (p6) {$5$};
    \node[state, right of = p6] (p7) {$6$};
    \node[state, right of = p7] (p8) {$7$};
    \node[state, right of = p8] (p9) {$8$};

    \path
      ([xshift=-2.5em] p1.south-|p1) node[anchor= north east] (st) {$\st$};
    \path[anchor=center]
      (st-|p1) node {$s_0$}
      (st-|p2) node {$s_0$}
      (st-|p3) node {$s_0$}

      (st-|p4) node {$s_2$}
      (st-|p5) node {$s_3$}
      (st-|p6) node {$s_2$}
      (st-|p7) node {$s_4$}
      (st-|p8) node {$s_5$}
      (st-|p9) node {$s_5$};

    \path ([yshift=-1em] st.east) node[anchor = east] (lab) {$\labels$};
    \path[anchor=center]
      (lab-|p1) node {$p$}
      (lab-|p2) node {$p$}
      (lab-|p3) node {$p$}
      (lab-|p5) node {$r$}
      (lab-|p7) node {$r$}
      (lab-|p8) node {$q$}
      (lab-|p9) node {$q$};

\path ([yshift=-1em] lab.east) node[anchor = east] (ty) {$\type$};
    \path[anchor=center]
      (ty-|p1) node {$\mathtt{L}$}
      (ty-|p2) node {$\mathtt{R}$}
      (ty-|p3) node {$\mathtt{R}$}

      (ty-|p4) node {$\mathtt{L}$}
      (ty-|p5) node {$\mathtt{L}$}
      (ty-|p6) node {$\mathtt{R}$}
      (ty-|p7) node {$\mathtt{L}$}
      (ty-|p8) node {$\mathtt{R}$}
      (ty-|p9) node {$\mathtt{L}$};

    \path[->] (p1) edge (p2)
              (p2) edge["{$-2$}"] (p3)
              (p3) edge["{$-2$}"] (p4)
              (p4) edge["{$-2$}"] (p5)
              (p5) edge["{$+1$}"] (p6)
              (p6) edge["{$-2$}"] (p7)
              (p7) edge["{$\genfrac{}{}{0pt}{}{+1}{c<0}$}"] (p8)
              (p8) edge["{$\genfrac{}{}{0pt}{}{-2}{c<0}$}"] (p9)
              (p1) edge [loop above left] (p1)
              (p5) edge [bend right=60, "{$+1$}"'] (p4)
              (p7) edge [loop above left, "$+1$"'] (p7)
              (p9) edge[loop above left] node[font=\tiny, anchor = west]
              {$\genfrac{}{}{0pt}{}{-2}{c<0}$} (p9);

    \node[above of = p1, node distance=3em] (comp) {};

    \draw[very thin] (p1.west|-comp) -- node[above] {$P_0$} (p1.east|-comp)
      -- ++(0,-0.5ex) (p1.west|-comp) -- ++(0,-0.5ex);
    \draw[very thin] (p2.west|-comp) -- node[above] {$P_1$} (p2.east|-comp)
      -- ++(0,-0.5ex) (p2.west|-comp) -- ++(0,-0.5ex);
    \draw[very thin] (p3.west|-comp) -- node[above] {$P_2$} (p3.east|-comp)
      -- ++(0,-0.5ex) (p3.west|-comp) -- ++(0,-0.5ex);

    \draw[very thin] (p4.west|-comp) -- node[above] {$P_3$} (p5.east|-comp)
      -- ++(0,-0.5ex) (p4.west|-comp) -- ++(0,-0.5ex);
    \draw[very thin] (p6.west|-comp) -- node[above] {$P_4$} (p6.east|-comp)
      -- ++(0,-0.5ex) (p6.west|-comp) -- ++(0,-0.5ex);
     \draw[very thin] (p7.west|-comp) -- node[above] {$P_5$} (p7.east|-comp)
       -- ++(0,-0.5ex) (p7.west|-comp) -- ++(0,-0.5ex);
    \draw[very thin] (p8.west|-comp) -- node[above] {$P_6$} (p8.east|-comp)
      -- ++(0,-0.5ex) (p8.west|-comp) -- ++(0,-0.5ex);
    \draw[very thin] (p9.west|-comp) -- node[above] {$P_7$} (p9.east|-comp)
      -- ++(0,-0.5ex) (p9.west|-comp) -- ++(0,-0.5ex);
\end{tikzpicture}
 \end{center}
  \caption{An APS $\mathcal{P}=(P_0,…,P_7)$ of the flat Kripke structure in \cref{fig:fks}}
\label{fig:example-ps-1}
\end{figure}

We denote by $|\mathcal{P}|=|P_0…P_n|$ the size of $\mathcal{P}$ and use global indices $\ell\in[0,|\mathcal{P}|-1]$ to address the $(\ell+1)$-th augmented state in $P_0…P_n$, denoted $\mathcal{P}[\ell]$.
To distinguish these global indices from positions in arbitrary sequences, we refer to them as \emph{locations} of $\mathcal{P}$.
Moreover, $\loc_\mathcal{P}(k)=\lbrace \ell\mid|P_0P_1…P_{k-1}|\le\ell<|P_0P_1…P_k|\rbrace$
denotes for $0\le k\le n$ the set of locations belonging to component $P_k$ and for all locations $\ell\in\loc_\mathcal{P}(k)$ we denote the corresponding component index in $\mathcal{P}$ by $\comp_\mathcal{P}(\ell)=k$.
For example, in \cref{fig:example-ps-1} we have $\loc_\mathcal{P}(3)=\{3,4\}$ and $\comp_\mathcal{P}(6)=5$ because the seventh state of $\mathcal{P}$ belongs to $P_5$.
We extend the component projections for augmented states
to (sequences of) locations of $\mathcal{P}$ and write, e.g.,
$\st_\mathcal{P}(\ell_1\ell_2)$ for $\st(\mathcal{P}[\ell_1])\st(\mathcal{P}[\ell_2])$ and $\update_\mathcal{P}(\ell)$ for $\update(\mathcal{P}[\ell])$.

An APS $\mathcal{P}$ gives rise to a counter system $\CS(\mathcal{P}) = (Q, 0, C, \Delta)$ where $Q=\lbrace0,…,|P|-1\rbrace$, $C$ are the counters used in the augmented states of $\mathcal{P}$ and $\Delta$ consists of those transitions $(\ell,\update_\mathcal{P}(\ell),\guards_\mathcal{P}(\ell'),\ell')$ such that $0\le \ell'=\ell+1<|\mathcal{P}|$ or $\ell'<\ell$ and $\{\ell',\ell'+1,\ldots,\ell\}=\loc_\mathcal{P}(k)$ for some loop $P_k$.
Notice that the APS in \cref{fig:example-ps-1} is presented as its corresponding counter system.
Let $\succpos_\mathcal{P}(\ell)$ denote the set $\lbrace\ell'\in Q\mid\exists\mathbf{u},G:(\ell,\mathbf{u}, G, \ell')\in\Delta\rbrace$ of successors of $\ell$ in $\CS(\mathcal{P})$.
A \emph{run} of $\mathcal{P}$ is a run of $\CS(\mathcal{P})$ that visits each location $\ell\in S$ at least once.
The set of all runs of $\mathcal{P}$ is denoted $\Runs(\mathcal{P})$. As a consequence, a run visits the last loop infinitely often. We say that an APS $\mathcal{P}$ is non-empty iff $\Runs(\mathcal{P}) \neq \emptyset$. Since every run $\sigma \in\Runs(\mathcal{P})$ corresponds, by construction of $\mathcal{P}$, to a path $\st_\mathcal{P}(\rho)\in Q^\omega$ in $\mathcal{K}$ we define the satisfaction of an \fLTL formula $\varphi$ at position $i$ by $(\sigma,i) \models_\mathcal{P} \varphi$ iff $(\st_\mathcal{P}(\sigma),i) \models \varphi$.

Finally, notice that $\CS(\mathcal{P})$ is in fact a \emph{flat} counter system.
It is shown in \cite{DBLP:journals/iandc/DemriDS15} that \LTL properties can be verified over flat counter systems in non-deterministic polynomial time.
Since \LTL can express that each location of $\CS(\mathcal{P})$ is visited we obtain the following result.

\begin{Lemma}[\cite{DBLP:journals/iandc/DemriDS15}]
\label{lem:flat-mc-np}
  Deciding non-emptiness of APS is in \textsc{NP}.
\end{Lemma}
 
\subsection{Labellings of Consistent APS are Correct}
\label{ssc:soundness}

An APS $\mathcal{P}$ assigns to every position $i$ on each of its runs $\sigma$ the labelling $L_i=\labels_\mathcal{P}(\sigma(i))$.
We are interested in this labelling being \emph{correct} with respect to some \fLTL formula $\Phi$ in the sense that $\Phi\in L_i$ if and only if $(\sigma,i) \models \Phi$.
The notion of consistency introduced in the following provides a sufficient criterion for correctness of the labelling of all runs of an APS.

An augmented path $u=a_0…a_n$  is said to be \emph{good}, \emph{neutral} or \emph{bad} for an \fLTL formula $\Psi=\varphi\U^\frac{x}{y}\psi$ if the number $d=|\lbrace0\le i<|u|\mid\varphi\in\labels(u(i))\rbrace|$ of positions labelled with $\varphi$ is larger than ($d>\frac{x}{y}\cdot|u|$), equal to ($d = \frac{x}{y}\cdot |u| $) or smaller than ($d<\frac{x}{y}\cdot|u|$), respectively, the fraction $\frac{x}{y}$ of all positions of $u$.
A tuple $(P_0,…,P_n)$ of rows and loops (not necessarily an APS) is called \emph{$L$-periodic} for a set $L\subseteq\sub(\Phi)$ of labels if all augmented paths $P_k$ share the same labelling with respect to $L$, that is for all $0\le k<n-1$ we have $|P_k|=|P_{k+1}|$ and $\labels(P_k(i))\cap L = \labels(P_{k+1}(i))\cap L$ for all $0\le i<|P_k|$.

\begin{Definition}[Consistency]
\label{def:consistency}
Let $\mathcal{P}=(P_0,…,P_n)$ be an APS in $\mathcal{K}$, $k\in[0,n]$ and $\ell\in\loc_\mathcal{P}(k)$ a location on component $P_k$.
The location $\ell$ is \emph{consistent} with respect to an \fLTL formula $\Psi$ if all locations of $\mathcal{P}$ are consistent with respect to all strict subformulae of $\Psi$ and one of the following conditions applies.

\begin{enumerate}
  \item \label{itm:consistency-bool}
    $\Psi\in AP$ and $\Psi\in\labels_\mathcal{P}(\ell) \Leftrightarrow \Psi\in\lambda(\st_\mathcal{P}(\ell))$, or
    $\Psi=\varphi\land\psi$ and $\Psi\in\labels_\mathcal{P}(\ell) \Leftrightarrow \varphi,\psi\in\labels_\mathcal{P}(\ell)$, or
  $\Psi=\neg\varphi$ and $\Psi\in\labels_\mathcal{P}(\ell) \Leftrightarrow \varphi\not\in\labels_\mathcal{P}(\ell)$.
  \item \label{itm:consistency-next}
    $\Psi=\X\varphi$ and $\forall \ell'\in\succpos_\mathcal{P}(\ell):  \Psi\in\labels_\mathcal{P}(\ell) \Leftrightarrow \varphi\in\labels_\mathcal{P}(\ell')$.
  \item \label{itm:consistency-until}
    $\Psi=\varphi\U^\frac{x}{y}\psi$ and one of the following holds:
    \begin{enumerate}
      \item \label{itm:consistency-psi} $\Psi,\psi\in\labels_\mathcal{P}(\ell)$
      \item \label{itm:consistency-final} $\Psi\in\labels_\mathcal{P}(\ell)$ and $P_n$ is good for $\Psi$ and $\exists\ell'\in\loc_\mathcal{P}(n): \psi\in\labels_\mathcal{P}(\ell')$
      \item \label{itm:consistency-counter} $\type_\mathcal{P}(\ell)=\mathtt{R}$ and there is a counter $c\in C$ such that
                                $\forall\ell'<\ell: \update_\mathcal{P}(\ell')(c) = 0$ and
        $\forall\ell'\ge \ell: \varphi\in\labels_\mathcal{P}(\ell') \Rightarrow \update_\mathcal{P}(\ell')(c) = y-x$ and
        $\forall\ell'\ge \ell: \varphi\not\in\labels_\mathcal{P}(\ell') \Rightarrow \update_\mathcal{P}(\ell')(c) = -x$  and
        \begin{itemize}
          \item if $\Psi\not\in\labels_\mathcal{P}(\ell)$ then $\psi\not\in\labels_\mathcal{P}(\ell)$ and $\forall\ell'>\ell: \psi\in\labels_\mathcal{P}(\ell')\Rightarrow (c<0)\in\guards_\mathcal{P}(\ell')$ and
          \item if $\Psi\in\labels_\mathcal{P}(\ell)$ then $\exists\ell'>\ell: \psi\in\labels_\mathcal{P}(\ell') \land (c\ge0)\in\guards_\mathcal{P}(\ell')$.
        \end{itemize}
      \item \label{itm:consistency-indirect} There is $k'\in[0,n]$ such that all locations $\ell'\in\loc_\mathcal{P}(k')$ are consistent wrt.\ $\Psi$ and
\begin{itemize}
  \item if $k=n$ then $k'<k$ and $(P_{k'},P_{k'+1},…,P_{k})$ is $\lbrace\varphi,\psi,\Psi\rbrace$-periodic,
  \item if $k<n$ and $P_k$ is good or neutral for $\Psi$ and $\Psi\not\in\labels_\mathcal{P}(\ell)$, or $P_k$ is bad for $\Psi$ and $\Psi\in\labels_\mathcal{P}(\ell)$ then $k'<k<n$ and $(P_{k'},P_{k'+1},…,P_{k+1})$ is $\lbrace\varphi,\psi,\Psi\rbrace$-periodic, and
  \item if $k<n$ and $P_k$ is good or neutral for $\Psi$ and $\Psi\in\labels_\mathcal{P}(\ell)$, or $P_k$ is bad for $\Psi$ and $\Psi\not\in\labels_\mathcal{P}(\ell)$ then $k<k'<n$ and $(P_k,P_{k+1},…,P_{k'+1})$ is $\lbrace\varphi,\psi,\Psi\rbrace$-periodic.
\end{itemize}
    \end{enumerate}
\end{enumerate}
The APS $\mathcal{P}$ is consistent with respect to $\Psi$ if it is the case for all its locations.
\end{Definition}

The cases~\lipicsbf{\ref{itm:consistency-bool}} and~\lipicsbf{\ref{itm:consistency-next}} reflect the semantics syntactically.
For instance, location $0$ in \cref{fig:example-ps-1} can be labelled consistently with $\X p$ since all its sucessor ($0$ and $1$) are labelled with $p$.
Case~\lipicsbf{\ref{itm:consistency-until}}, concerning the (frequency) until operator, is more involved.

Assume that $\Phi=\varphi\U^\frac{x}{y}\psi$ is an until formula and that the labelling of $\mathcal{K}$ by $\varphi$ and $\psi$ is consistent.
In some cases, it is obvious that $\Phi$ holds, namely at positions labelled by $\psi$ (case~\lipicsbf{\ref{itm:consistency-psi}}) or if the final loop already guarantees that $\Phi$ always holds (case~\lipicsbf{\ref{itm:consistency-final}}).
If neither is the case we can apply the idea discussed in \cref{sec:fctl} and use a counter to check explicitly if at some point the formula $\Phi$ holds (case~\lipicsbf{\ref{itm:consistency-counter}}).
Recall that to validate (or invalidate) the labelling of a location by the formula $\Phi$ a specific counter tracks the frequency constraint in terms of the balance between fees and rewards along a run.
For the starting point to be unique this case only applies to locations that are not part of a loop.
For those labelled with $\Phi$ there should exist a location in the future where $\psi$ holds and the balance counter is non-negative.
For those not labelled with $\Phi$ all locations in the future where $\psi$ holds must be entered with negative balance.
Finally, case~\lipicsbf{\ref{itm:consistency-indirect}} can apply (not only) to loops and is based on the following reasoning:
if a loop is good (bad) and $\Phi$ is supposed to hold at some of its locations then it suffices to verify that this is the case during any of its future (past) iterations, e.g.\ the last (first) and vice versa if $\Phi$ is supposed not to hold.
This is the reason why this case allows for delegating consistency along a periodic pattern.

For instance, consider the formula $\Psi=r\U^{\frac{2}{3}}q$ and the APS shown in \cref{fig:example-ps-1}.
It is consistent to \emph{not} label location $1$ by $\Psi$ because the counter $c$ tracks the balance and locations $7$ and $8$ are guarded as required.
If a run takes, e.g., the loop $P_5$ seven times, it has to take $P_3$ at least twice to satisfy all guards.
This ensures that the ratio for the proposition $r$ is strictly less than $\frac{2}{3}$ upon reaching the first (and thus any) occurrence of $q$.
Note that to also make location $2$ consistent, an additional counter needs to be added.
Consistency with respect to $\Psi$ is then inherited by location $0$ from location $1$ according to case~\lipicsbf{\ref{itm:consistency-indirect}} of the definition.
Intuitively, additional iterations of the bad loop $P_0$ can only diminish the ratio.

The definition of consistency guarantees that if an APS is consistent with respect to $\Phi$ then for every run of the APS, each time the formula $\Phi$ is encountered, it holds at the current position (see \cref{apx:correctness} for complete details).
Hence we obtain the following lemma that guarantees correctness of our decision procedure.

\begin{restatable}[Correctness]{Lemma}{fltlsoundness}
\label{lem:fltl-soundness}
  If there is an APS $\mathcal{P}$ in $\mathcal{K}$ such that $\mathcal{P}$ is consistent wrt.\ $\Phi$ and $\Phi\in\labels_\mathcal{P}(0)$ and $\Runs(\mathcal{P})\neq\emptyset$ then $\mathcal{K}\models\Phi$.
\end{restatable}
 
\subsection{Constructing Consistent APS}
\label{ssc:completeness}
Assuming that our flat Kripke structure $\mathcal{K}$ admits a run $\rho$ such that $\rho\models\Phi$, we show how to construct a non-empty APS that is initially labelled by and consistent with respect to $\Phi$.
It will be of at most exponential size in $|\mathcal{K}|+|\Phi|$ and is built recursively over the structure of $\Phi$.

Concerning the base case where $\Phi\in AP$, all paths in a flat structure can be represented by a path schema of linear size \cite{DBLP:conf/atva/LerouxS05,DBLP:journals/iandc/DemriDS15}.
Intuitively, since $\mathcal{K}$ is flat, every subpath $s_is_{i+1}…s_{i'}…s_{i''}$ of $\rho$ where a state $s_i=s_{i'}=s_{i''}$ occurs more than twice is equal to $(s_is_{i+1}…s_{i'-1})^ks_{i''}$ for some $k \in \nat$.
Hence, there are simple subpaths $u_0,…,u_m\in S^+$ of $\rho$ and positive numbers of iterations $n_0,…,n_{m-1}\in\mathbb{N}$ such that $\rho=u_0^{n_0}u_1^{n_1}…u_{m-1}^{n_{m-1}}u_m^\omega$ and $|u_0u_1…u_m|\le2|S|$.
From this decomposition, we build an APS being consistent with respect to all propositions.
Henceforth, we assume by induction an APS $\mathcal{P}$ being consistent with respect to all strict subformulae of $\Phi$ and a run $\sigma\in\Runs(\mathcal{P})$ with $\st_\mathcal{P}(\sigma)=\rho$.
If $\Phi=\varphi\land\psi$ or $\Phi=\neg\varphi$, \cref{def:consistency} determines for each augmented state of $\mathcal{P}$ whether it is supposed to be labelled by $\Phi$ or not.
It remains hence to deal with the next and frequency until operators.

\subparagraph{Labelling $\mathcal{P}$ by $\X\varphi$.}
If $\Phi=\X\varphi$ the labelling at some location $\ell$ is extended according to the labelling of its successors.
These may disagree upon $\varphi$ (only) if $\ell$ has more than one successor, i.e., being the last location on a loop $P_k$ of $\mathcal{P}=(P_0,…,P_m)$.
In that case we consult the run $\sigma$: if it takes $P_k$ only once, this loop can be \emph{cut} and replaced by $P_k'$ that we define to be an exact copy except that all augmented states have type $\mathtt{R}$ instead of $\mathtt{L}$.
If otherwise $\sigma$ takes $P_k$ at least twice, the loop can be \emph{unfolded} by inserting $P_k'$ between $P_k$ and $P_{k+1}$, i.e.\ letting $\mathcal{P}'=(P_0,…,P_k,P_k',P_{k+1},…,P_m)$.
Either way, $\sigma$ remains a run of the obtained APS, up to shifting the locations $\ell'>\ell$ if the extra component was inserted (recall that locations are indices).
Importantly, cutting or unfolding any loop, even any number of times, in $\mathcal{P}$ preserves consistency.

\subparagraph{Labelling $\mathcal{P}$ by $\varphi\U^r\psi$.}
The most involved case is to label a location $\ell$ by $\Phi=\varphi\U^r\psi$.
First, assume that $\ell$ is part of a row.
Whether it must be labelled by $\Phi$ is uniquely determined by $\sigma$.
This is consistent if case~\lipicsbf{\ref{itm:consistency-psi}} or~\lipicsbf{\ref{itm:consistency-final}} of \cref{def:consistency} applies.
The conditions of case~\lipicsbf{\ref{itm:consistency-counter}} are also realised easily in most situations.
Only, if $\Phi$ holds at $\ell$ but every location $\ell'$ witnessing this (by being reachable with sufficient frequency and labelled by $\psi$) is part of some loop $P'$.
Adding the required guard directly to $\ell'$ may be too strict if $\sigma$ traverses $P'$ more than once.
However, the first iteration (if $P'$ is bad for $\Phi$) or the last iteration (if $P'$ is good) on $\sigma$ contains a position (labelled with $\psi$) witnessing that $\Phi$ holds if any iteration does.
Thus it suffices to unfold the loop once in the respective direction.
For example, consider in \cref{fig:example-ps-1} location 5 and a formula $\varphi=r\U^{\frac{2}{5}}q$.
Location 8 could witness that  $\varphi$ holds but a corresponding guard would be violated eventually since $P_7$ is bad for $\varphi$.
The first iteration is thus the optimal choice.
The unfolding $P_6$ separates it such that location 7 can be guarded instead without imposing unnecessary constraints.

Now assume that location $\ell$, to be labelled or not with $\Phi$, is part of a loop $P$ which is \emph{stable} in the sense that $\Phi$ holds either at all positions $i$ with $\sigma(i)=\ell$ or at none of them.
With two unfoldings of $P$, made consistent as above, case~\lipicsbf{\ref{itm:consistency-indirect}} applies.
However, $\sigma$ may go through $\ell$ several, say $n>1$, times where $\Phi$ holds at some but not all of the corresponding positions.
If $n$ is small we can replace $P$ by precisely $n$ unfoldings, thus reducing to the previous case without increasing the size of the structure too much.
We can moreover show that if $n$ is not small then it is possible to decompose such a problematic loop into a constant number of unfoldings and two stable copies based on the following observation.

\begin{restatable}[Decomposition]{Lemma}{fltldecomposition}
\label{lem:fltl-decomposition}
Let $P=\mathcal{P}[\ell_0]…\mathcal{P}[\ell_{|P|-1}]$ be a non-terminal loop in $\mathcal{P}$ with corresponding location sequence $v=\ell_0…\ell_{|P|-1}$ and $\hat{n}={|P|}\cdot y$ for some $y>0$.
For every run $\sigma=uv^nw\in\Runs(\mathcal{P})$ where $n\ge\hat{n}+2$ there are $n_1$ and $n_2$ such that $\sigma=uv^{n_1}v^{\hat{n}}v^{n_2}w$ and for all positions $i$ on $\sigma$ with $|u| \le i < |uv^{n_1-1}|$ or $|uv^{n_1}v^{\hat{n}}|\le i<|uv^{n_1}v^{\hat{n}}v^{n_2-2}|$ we have $(\sigma,i)\models_\mathcal{P}\Phi$ iff $(\sigma,i+|P|)\models_\mathcal{P}\Phi$.
\end{restatable}

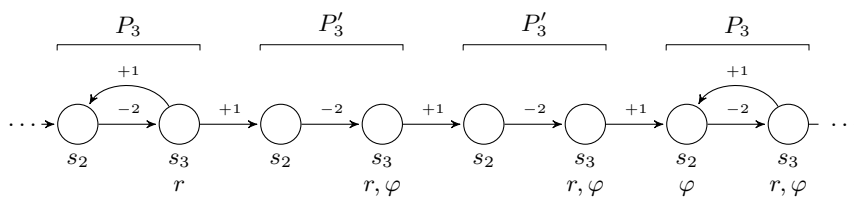
\begin{figure}[tb]
\begin{center}

\begin{tikzpicture}[small nodes]
    \node[inner sep = 0pt] (begin) {$…$};
    \node[state, right of = begin, node distance = 2em] (p1) {};    \node[state, right of = p1] (p2) {};    \node[state, right of = p2] (p3) {};    \node[state, right of = p3] (p4) {};    \node[state, right of = p4] (p5) {};    \node[state, right of = p5] (p6) {};    \node[state, right of = p6] (p7) {};    \node[state, right of = p7] (p8) {};    \node[right of = p8, node distance = 2.2em] (end) {…};

    \path (p1.south) node (st) {};
    \path[anchor=north]
      (st-|p1) node {$s_2$}
      (st-|p2) node {$s_3$}
      (st-|p3) node {$s_2$}
      (st-|p4) node {$s_3$}
      (st-|p5) node {$s_2$}
      (st-|p6) node {$s_3$}
      (st-|p7) node {$s_2$}
      (st-|p8) node {$s_3$};

    \path (st) ++ (0, -1em) node (lab) {};
    \path[anchor=north]
      (lab-|p1) node {}
      (lab-|p2) node {$r$}
      (lab-|p3) node {}
      (lab-|p4) node {$r,\varphi$}
      (lab-|p5) node {}
      (lab-|p6) node {$r,\varphi$}
      (lab-|p7) node {$\varphi$}
      (lab-|p8) node {$r,\varphi$};

    \path[->] (begin) edge (p1)
              (p1) edge["$-2$"] (p2)
              (p2) edge["$+1$"] (p3)
              (p3) edge["$-2$"] (p4)
              (p4) edge["$+1$"] (p5)
              (p5) edge["$-2$"] (p6)
              (p6) edge["$+1$"] (p7)
              (p7) edge["$-2$"] (p8)
              (p8) edge[-=] (end)
              (p2) edge[bend right = 60, "$+1$"'] (p1)
              (p8) edge[bend right = 60, "$+1$"'] (p7);

  \node[above of = p1, node distance=3em] (comp) {};

    \draw[very thin] (p1.west|-comp)
        -- node[above] (comp1) {$P_3$} (p2.east|-comp)
        -- ++(0,-0.5ex)
      (p1.west|-comp) -- ++(0,-0.5ex);
  \draw[very thin] (p3.west|-comp) -- node[above] (comp2) {$P_3'$} (p4.east|-comp)
    -- ++(0,-0.5ex) (p3.west|-comp) -- ++(0,-0.5ex);
  \draw[very thin] (p5.west|-comp) -- node[above] (comp3) {$P_3'$} (p6.east|-comp)
    -- ++(0,-0.5ex) (p5.west|-comp) -- ++(0,-0.5ex);
  \draw[very thin] (p7.west|-comp) -- node[above] (comp4) {$P_3$} (p8.east|-comp)
    -- ++(0,-0.5ex) (p7.west|-comp) -- ++(0,-0.5ex);

\end{tikzpicture}
 \end{center}
  \caption{A decomposition of loop $P_3$ from \cref{fig:example-ps-1} allowing for a correct labelling wrt.\ $\varphi=r\U^{\frac{2}{3}} q$.\label{fig:example-ps-2}}
\end{figure}

\begin{Example}
Consider again the APS $\mathcal{P}$ in \cref{fig:example-ps-1}, a run $\sigma\in\Runs(\mathcal{P})$ and the location $3$.
Whether or not $\varphi=r\U^{\frac{2}{3}}q$ holds at some position $i$ with $\sigma(i)=3$ depends on how often $\sigma$ traverses the good loop $P_5$ (the more the better) and how often it repeats $P_3$ after position $i$ (the more the worse).
Assume $\sigma$ traverses $P_5$ exactly five times and $P_3$ sufficiently often, say 10 times.
Then, during the last three iterations of $P_3$, $\varphi$ holds when visiting location $3$, and also location $4$.
In the two iterations before, the formula holds exclusively at location $4$ and in any preceding iteration, it does not hold at all.
Thus any labelling of $P_3$ would necessarily be incorrect.
However, we can replace $P_3$ by four copies of it that are labelled as indicated in \cref{fig:example-ps-2} and $\sigma$ can easily be mapped onto this modified structure.
\end{Example}

The presented procedure for constructing an APS from the run $\rho$ in $\mathcal{K}$ performs only linearly many steps in $|\Phi|$, namely one step for each subformula.
It starts with a structure of size at most $2|\mathcal{K}|$ and all modifications required to label an APS increase its size by a constant factor.
Hence, we obtain an APS $\mathcal{P}_\Phi$ of size at most exponential in the length of $\Phi$ and polynomial in the number of states of $\mathcal{K}$.
This consistent APS still contains a run corresponding to $\rho$ and hence its first location must be labelled by $\Phi$ because $(\rho,0)\models\Phi$ and we have seen that consistency implies correctness.

\begin{restatable}[Completeness]{Lemma}{fltlcompleteness}
\label{lem:fltl-completeness}
  If $\mathcal{K} \models \Phi$ then there is a consistent APS $\mathcal{P}$ in $\mathcal{K}$ of at most exponential size in $\mathcal{K}$ and $\Phi$ where $\Phi\in\labels(\mathcal{P}(0))$ and $\mathcal{P}$ is non-empty.
\end{restatable}

We have seen in this section that the decision procedure presented in the beginning is sound and complete due to \cref{lem:fltl-soundness} and~\ref{lem:fltl-completeness}, respectively.
The guessed APS is of exponential size in $|\Phi|$ and of polynomial size in $|\mathcal{K}|$.
Since both checking consistency and non-emptiness (cf.\ \cref{lem:flat-mc-np}) require polynomial time (in the size of the APS) the procedure requires at most exponential time.

\begin{Theorem}
 $\MC{\FKS,\fLTL}$   is in \textsc{NExp}.
\end{Theorem}

This result immediately extends to \fCTLstar.
For a state $q$ of a flat Kripke structure $\mathcal{K}$ and an arbitrary \fLTL formula $\varphi$, the procedure allows us to decide in \textsc{NExp} whether $q\models\E \varphi$ holds.
It allows us further to decide if $q\models\A \varphi$ holds in \textsc{ExpSpace} by the dual formulation $q\not\models\E\neg\varphi$ and Savitch's theorem.
Following otherwise the standard labeling procedure for \CTL (cf.\ \cref{sec:fctl}) requires to invoke the procedure a polynomial number of times in $|\mathcal{K}|+|\Phi|$.

\begin{Theorem}
 $\MC{\FKS,\fCTLstar}$ is in \textsc{ExpSpace}.
\end{Theorem}

\section{On model-checking $\CCTLstar$ over flat Kripke structures}
\label{sec:mc-cctlstar}

In this section, we prove decidability of $\MC{\FKS,\CCTLstar}$.
We provide a polynomial encoding into the satisfiability problem of a decidable extension of Presburger arithmetic featuring a quantifier for counting the solutions of a formula.
For the reverse direction an exponential reduction provides a corresponding hardness result for \CLTL, \CCTL and \CCTLstar.

\subparagraph{Presburger arithmetic with H\"artig quantifier.}
\label{subsec:PA}
First-order logic over the natural numbers with addition was shown to be decidable by M.\ Presburger \cite{Presburger29}.
It has been extended with the so-called \emph{H\"artig quantifier}~\cite{ape66,DBLP:conf/pldi/Pugh94,DBLP:journals/tocl/Schweikardt05} that allows for referring to the number of values for a specific variable that satisfy a formula.
We denote this extension by \PH.
The syntax of \PH formulae $\varphi$ and \PH terms $\tau$ over a set of variables $V$ is defined by the grammar
\begin{align*}
  \varphi & ::= \tau \le \tau \mid \neg\varphi \mid \varphi \land \varphi \mid \exists x.\varphi \mid\exists^{=x}y.\varphi
&
  \tau & ::=  a \mid a \cdot x \mid \tau + \tau
\end{align*}
for natural constants $a\in\mathbb{N}$ and variables $x,y \in V$.
Since the structure $(\mathbb{N},+)$ is fixed, the semantics is defined over valuations $\eta:V\to\mathbb{N}$ that are extended to terms $t$ as expected, e.g., $\eta(3 \cdot x+1) = 3\cdot\eta(x)+1$.
We define the satisfaction relation $\models_\PH$ as usual for first-order logics and by
$
  \eta \models_\PH \exists^{=x}y.\varphi
 ~~\defeq ~~
     \mathbb{N} \ni |\set{b \in \mathbb{N}\mid\eta[y \mapsto b] \models_\PH \varphi}| = \eta(x)
$
for the Härtig quantifier.
Notice that the solution set has to be finite.

The \emph{satisfiability problem} of $\PresHartig$ consists in
determining whether for a \PH formula $\varphi$  there exists a valuation
$\eta$ such that $\eta\models_{\PH} \varphi$.
It is decidable \cite{ape66,DBLP:conf/pldi/Pugh94,DBLP:journals/tocl/Schweikardt05} via eliminating the Härtig quantifier, but its complexity is not known.
For what concerns classic Presburger arithmetic, the complexity of its satisfiability problem lies between \textsc{2Exp} and \textsc{2ExpSpace} \cite{DBLP:journals/tcs/Berman80}.
 
\subparagraph{Lower bound for $\MC{\FKS,\CCTLstar}$.}
Let $\mathcal{K}$ be the flat Kripke structure over $AP=\emptyset$ that consists of a single loop of length one.
We can encode satisfiability of a \PH formula $\Phi$ into the question whether the (unique) run $\rho$ of $\mathcal{K}$ satisfies a \CLTL formula $\hat{\Phi}$.
Assume without loss of generality that $\Phi$ has no free variables.
Let $V_\Phi$ be the variables used in $\Phi$ and $z_1,z_2,…\not\in V_\Phi$ additional variables.
Recall that $\rho\models\hat{\Phi}$ if $(\rho,\theta,0)\models\hat{\Phi}$ for some valuation $\theta$ of the positional variables in $\hat{\Phi}$.

The idea is essentially to encode the value given to a variable $x\in V_\Phi$ of $\Phi$ into the distance between the positions assigned to two variables of $\hat{\Phi}$.
Technically, a mapping $Z\in\mathbb{N}^{V_\Phi}$ associates with each variable $x\in V_\Phi$ an index $j=Z(x)$ and the constraints that $\Phi$ imposes on $x$ are translated to constraints on positional variables $z_j$ and $z_{j-1}$ (more precisely, the distance $\theta(z_j) - \theta(z_{j-1})$ between the assigned positions).
The following transformation $\mathtt{t}: \PH \times \mathbb{N}^{V_\Phi} \times \mathbb{N} \to \CLTL$ constructs the \CLTL formula from $\Phi$.
When a variable is encountered, the mapping $Z$ is updated by assigning to it the next free index (third parameter).
Let
\[\begin{array}{lcl@{\hspace{5em}}lcl}
    \mathtt{t}(\varphi_1 \odot \varphi_2, Z, i) &=& \mathtt{t}(\varphi_1,Z, i) \odot \mathtt{t}(\varphi_2,Z, i) & \mathtt{t}(\neg\varphi, Z, i)      &=& \neg \mathtt{t}(\varphi,Z, i)\\
    \mathtt{t}(a \cdot x,Z,i)           &=& a \cdot \#_{z_{Z(x)-1}}(\true) - a \cdot \#_{z_{Z(x)}}(\true) & \mathtt{t}(a, Z, i)       &=& a\\
  \mathtt{t}(\exists x.\varphi, Z, i)     &=& \F z_i.\mathtt{t}(\varphi, Z[x\mapsto i], i+1) \\
          \mathtt{t}(\exists^{=x}y.\varphi, Z, i) &=& \multicolumn{4}{l}{
    \F\G \left(\mathtt{t}(x,Z,i) = \#_{z_{i-1}}(z_i.\mathtt{t}(\varphi,Z[y\mapsto i],i+1))\right)
  }\end{array}\]
for $x,y\in V_\Phi$, $a,i\in\mathbb{N}$ and $\odot\in\lbrace\land,\le,+\rbrace$.
Then, we obtain $\hat{\Phi}=z_0.\mathtt{t}(\Phi,\mathbf{1},1)$, initialising $Z$ and the first free index with $1$.
Notice that the translation of the Härtig quantifier instantiates the scope effectively twice when substituting the equality and thus the size of $\hat{\Phi}$ may at worst double with each nesting.
Finally, we can equivalently add path quantifiers to all temporal operators in $\hat{\Phi}$ and obtain, syntactically, a \CCTL formula.

\begin{Theorem}\label{thm:ph-to-cltl}
  The satisfiability problem of \PH is reducible in exponential time to both
   $\MC{\FKS,\CLTL}$ and $\MC{\FKS,\CCTL}$.
\end{Theorem}
 
\subparagraph{Deciding $\MC{\FKS,\CCTLstar}$.}

We provide a polynomial reduction to the satisfiability problem of $\PresHartig$.
Given a flat Kripke structure $\mathcal{K}$ we can represent each run $\rho$ by a fixed number of naturals.
We use a predicate $Conf$ that allows for accessing the $i$-th state on $\rho$ given its encoding and a predicate $Run$ characterising all (encodings of) runs in $\Runs(\mathcal{K})$.
Such predicates were shown to be definable by Presburger arithmetic formulae of polynomial size and used to encode \MC{\FKS,\CTLstar} \cite{DBLP:journals/jancl/DemriFGD10,DBLP:conf/rp/DemriDS14}.
We adopt this idea for $\MC{\FKS,\CCTLstar}$ and $\PresHartig$.
Let $\mathcal{K}=(S,s_I, E, \lambda)$ and assume $S \subseteq \nat$ without loss of generality.
For $N\in\mathbb{N}$ let $V_N=\lbrace r_1,…,r_N,i,s\rbrace$ be a set of variables that we use to encode a run, a position and a state, respectively.

\begin{Lemma}[\cite{DBLP:conf/rp/DemriDS14}]
  There is a number $N\in\mathbb{N}$, a mapping $enc:\mathbb{N}^N\to
  S^\omega$ and predicates $Conf(r_1,…,r_N,i,s)$ and $Run(r_1,…,r_N)$
  such that for all valuations $\eta:V_N\to\mathbb{N}$ we have
  \begin{inparaenum}
    \item $\eta \models_\PH Run(r_1,…,r_N) \ \Leftrightarrow \ enc(\eta(r_1),…,\eta(r_N))\in\Runs(\mathcal{K})$ and
    \item if $\eta \models_\PH Run(r_1,…,r_N)$ then $\eta\models_\PH Conf(r_1,…,r_N,i,s)\ \Leftrightarrow \ enc(\eta(r_1),…,\eta(r_N))(\eta(i)) = \eta(s)$.
  \end{inparaenum}
Both predicates are definable by \PH formulae over variables $V\supseteq V_N$ of polynomial size in $|\mathcal{K}|$.
\end{Lemma}

Now, let $\Phi$ be a \CCTLstar formula to be verified on $\mathcal{K}$.
Without loss of generality we assume that all comparisons $\varphi_\le\in\sub(\Phi)$ of the form $\tau_1\le\tau_2$ have the shape $\varphi_\le = \sum_{\ell=1}^{k}a_\ell\cdot\#_{x_\ell}(\varphi_\ell) + b \le \sum_{\ell=k+1}^{m}a_\ell\cdot\#_{x_\ell}(\varphi_\ell) + c$ for some $k,m,b,c\in\mathbb{N}$, coefficients $a_\ell\in\mathbb{N}$ and subformulae $\varphi_\ell$.
As it is done in \cite{DBLP:conf/rp/DemriDS14} for \CTL, using the predicates $Conf$ and $Run$, we construct a \PH formula that is satisfiable if and only if $\mathcal{K}\models\Phi$.
Given the encoding of relevant runs into natural numbers we can express path quantifiers with quantification over the variables $r_1,…,r_N$.
Temporal operators can be expressed by using $Conf$ to access specific positions.
Storing of positions is done explicitly by assigning them as value to specific variables $x$.
Variables $z$ are introduced to hold the number of positions satisfying a formula and can then be used in constraints.
For example, to translate a term $\#_x(\varphi)$ we specify a variable, e.g., $z_1$ holding this value by $\exists z_1.\exists^{=z_1}i'.x\le i'\le i \land \hat{\varphi}$ where $i$ holds the current position and $\hat{\varphi}$ expresses that $\varphi$ holds at position $i'$ of the current run.
Constraints like $\#_x(\varphi)+1\le\#_x(\psi)$ can now directly be translated to, e.g., $z_1 +1 \le z_2$.
We use a syntactic translation function $\chk$ that takes the formula $\varphi$ to be translated, the names of $N$ variables encoding the current run and the name of the variable holding the current position.
Let
\[\begin{array}{l@{\,}c@{\,}l}
  \chk(p,r_1,…,r_N,i) &=&
     \exists s.Conf(r_1,…,r_N,i,s) \land \bigvee_{a\mid p\in\lambda(a)} s = a\\
  \chk(\varphi\land\psi,r_1,…,r_N,i) &=& \chk(\varphi,r_1,…,r_N,i) \land \chk(\psi,r_1,…,r_N,i)\\
  \chk(\neg\varphi,r_1,…,r_N,i) &=& \neg\chk(\varphi,r_1,…,r_N,i)\\
  \chk(\X \varphi,r_1,…,r_N,i) &=&
    \exists i'.i' = i+1 \land \chk(\varphi,r_1,…,r_N,i') \\
  \chk(\varphi \U \psi, r_1, …, r_N, i) &=&
    \exists i''.  i\le i'' \land \chk(\psi,r_1,…,r_N,i'')\ \land \\
    && \forall i'.(i\le i' \land i'<i'') \to \chk(\psi,r_1,…,r_N,i') \\
  \chk(\E \varphi,r_1,…,r_N,i) &=&
    \exists r'_1 … \exists r'_N. Run(r'_1,…,r'_N) \land \chk(\varphi,r'_1,…,r'_N,i) \land \forall i'. \\
    && (i'\le i) \to \exists s. Conf(r_1,…,r_N,i',s)\land Conf(r'_1,…,r'_N,i',s)\\
  \chk(x.\varphi,r_1,…,r_N,i) &=& \exists x.x=i \land \chk(\varphi,r_1,…,r_N,i)\\
  \chk(\varphi_\le,r_1,…,r_N,i) &=& \exists z_1…\exists z_m.
    \left(\bigwedge_{\ell=1}^{m}\exists^{=z_\ell}i'.x_\ell\le i'\le i \land \chk(\varphi_\ell,r_1,…,r_N,i')\right)\\
    && \land\ a_1 \cdot z_1 + … + a_k \cdot z_k + b \le a_{k+1}\cdot z_{k+1} + … + a_m\cdot z_m + c
\end{array}\]
for $\varphi_\le = \sum_{\ell=1}^{k}a_\ell\cdot\#_{x_\ell}(\varphi_\ell) + b \le \sum_{\ell=k+1}^{m}a_\ell\cdot\#_{x_\ell}(\varphi_\ell) + c$.
Primed variables denote fresh copies of the corresponding input variables, e.g.\ $i'$ becomes $(i')'=i''$ and $i''$ becomes $i'''$.
Now, $\Phi\models\mathcal{K}$ if and only if
$
\exists r_1…\exists r_N.\exists i.Run(r_1,…,r_N) \land i = 0 \land \chk(\Phi,r_1,…,r_N,i)
$
is satisfiable.

\begin{theorem}
$\MC{\FKS,\CCTLstar}$ is reducible to \PH satisfiability in polynomial time.
\end{theorem}

\section{Conclusion}

In this paper, we have seen that model checking flat Kripke structures with some expressive counting temporal logics is possible whereas this is not the case for general, finite Kripke structures.
However, our results provide an under-approximation approach to this latter problem that consists in constructing flat sub-systems of the considered Kripke structure.
We furthermore believe our method works as well for flat counter systems.
We left as open problem the precise complexity for model checking \fCTL, \fLTL and \fCTLstar over flat Kripke structures.
It follows from \cite{DBLP:conf/concur/KuhtzF11} that the latter two problems are \textsc{NP}-hard while we obtain exponential upper bounds.
However, we believe that if we fix the nesting depth of the frequency until operator in the logic, the complexity could be improved.

This work has shown, as one could have expected, a strong connection between \CLTL and counter systems and as future work we plan to study automata-based formalisms inspired by \fLTL where we will equip our automata with some counters whose role will be to evaluate the relative frequency of particular events.

\bibliography{bibliography}

\begin{thebibliography}{10}

\bibitem{DBLP:conf/fsttcs/AbdullaAMS15}
Parosh~Aziz Abdulla, Mohamed~Faouzi Atig, Roland Meyer, and Mehdi~Seyed Salehi.
\newblock What's decidable about availability languages?
\newblock In Prahladh Harsha and G.~Ramalingam, editors, {\em 35th {IARCS}
  Annual Conference on Foundation of Software Technology and Theoretical
  Computer Science, {FSTTCS} 2015, December 16-18, 2015, Bangalore, India},
  volume~45 of {\em LIPIcs}, pages 192--205. Schloss Dagstuhl - Leibniz-Zentrum
  fuer Informatik, 2015.
\newblock \href {http://dx.doi.org/10.4230/LIPIcs.FSTTCS.2015.192}
  {\path{doi:10.4230/LIPIcs.FSTTCS.2015.192}}.

\bibitem{ape66}
H.~Apelt.
\newblock {Axiomatische Untersuchungen \"uber einige mit der Presburgerschen
  Arithmetik verwandten Systeme.}
\newblock {\em Z. Math. Logik Grundlagen Math.}, 12:131--168, 1966.

\bibitem{DBLP:books/daglib/0020348}
Christel Baier and Joost{-}Pieter Katoen.
\newblock {\em Principles of model checking}.
\newblock {MIT} Press, 2008.

\bibitem{DBLP:journals/tcs/Berman80}
Leonard Berman.
\newblock The complexitiy of logical theories.
\newblock {\em Theor. Comput. Sci.}, 11:71--77, 1980.
\newblock \href {http://dx.doi.org/10.1016/0304-3975(80)90037-7}
  {\path{doi:10.1016/0304-3975(80)90037-7}}.

\bibitem{DBLP:conf/tase/BolligDL12}
Benedikt Bollig, Normann Decker, and Martin Leucker.
\newblock Frequency linear-time temporal logic.
\newblock In Tiziana Margaria, Zongyan Qiu, and Hongli Yang, editors, {\em
  Sixth International Symposium on Theoretical Aspects of Software Engineering,
  {TASE} 2012, 4-6 July 2012, Beijing, China}, pages 85--92. {IEEE} Computer
  Society, 2012.
\newblock \href {http://dx.doi.org/10.1109/TASE.2012.43}
  {\path{doi:10.1109/TASE.2012.43}}.

\bibitem{DBLP:conf/concur/BouajjaniEM97}
Ahmed Bouajjani, Javier Esparza, and Oded Maler.
\newblock Reachability analysis of pushdown automata: Application to
  model-checking.
\newblock In Antoni~W. Mazurkiewicz and J{\'{o}}zef Winkowski, editors, {\em
  {CONCUR} '97: Concurrency Theory, 8th International Conference, Warsaw,
  Poland, July 1-4, 1997, Proceedings}, volume 1243 of {\em Lecture Notes in
  Computer Science}, pages 135--150. Springer, 1997.
\newblock \href {http://dx.doi.org/10.1007/3-540-63141-0\_10}
  {\path{doi:10.1007/3-540-63141-0\_10}}.

\bibitem{DBLP:conf/lop/ClarkeE81}
Edmund~M. Clarke and E.~Allen Emerson.
\newblock Design and synthesis of synchronization skeletons using
  branching-time temporal logic.
\newblock In Dexter Kozen, editor, {\em Logics of Programs, Workshop, Yorktown
  Heights, New York, May 1981}, volume 131 of {\em Lecture Notes in Computer
  Science}, pages 52--71. Springer, 1981.
\newblock \href {http://dx.doi.org/10.1007/BFb0025774}
  {\path{doi:10.1007/BFb0025774}}.

\bibitem{DBLP:journals/cacm/ClarkeES09}
Edmund~M. Clarke, E.~Allen Emerson, and Joseph Sifakis.
\newblock Model checking: algorithmic verification and debugging.
\newblock {\em Commun. {ACM}}, 52(11):74--84, 2009.
\newblock \href {http://dx.doi.org/10.1145/1592761.1592781}
  {\path{doi:10.1145/1592761.1592781}}.

\bibitem{DBLP:conf/concur/DeckerHLST17}
Normann Decker, Peter Habermehl, Martin Leucker, Arnaud Sangnier, and Daniel
  Thoma.
\newblock Model-checking counting temporal logics on flat structures.
\newblock In Roland Meyer and Uwe Nestmann, editors, {\em 28th International
  Conference on Concurrency Theory, {CONCUR} 2017}, volume~85 of {\em LIPIcs},
  pages 25:1--25:16. Schloss Dagstuhl - Leibniz-Zentrum fuer Informatik, 2017.
\newblock \href {http://dx.doi.org/10.4230/LIPIcs.CONCUR.2017.25}
  {\path{doi:10.4230/LIPIcs.CONCUR.2017.25}}.

\bibitem{DBLP:conf/rp/DemriDS14}
St{\'{e}}phane Demri, Amit~Kumar Dhar, and Arnaud Sangnier.
\newblock Equivalence between model-checking flat counter systems and
  {P}resburger arithmetic.
\newblock In Jo{\"{e}}l Ouaknine, Igor Potapov, and James Worrell, editors,
  {\em Reachability Problems - 8th International Workshop, {RP} 2014, Oxford,
  UK, September 22-24, 2014. Proceedings}, volume 8762 of {\em Lecture Notes in
  Computer Science}, pages 85--97. Springer, 2014.
\newblock \href {http://dx.doi.org/10.1007/978-3-319-11439-2\_7}
  {\path{doi:10.1007/978-3-319-11439-2\_7}}.

\bibitem{DBLP:journals/iandc/DemriDS15}
St{\'{e}}phane Demri, Amit~Kumar Dhar, and Arnaud Sangnier.
\newblock Taming past {LTL} and flat counter systems.
\newblock {\em Inf. Comput.}, 242:306--339, 2015.
\newblock \href {http://dx.doi.org/10.1016/j.ic.2015.03.007}
  {\path{doi:10.1016/j.ic.2015.03.007}}.

\bibitem{DBLP:journals/jancl/DemriFGD10}
St{\'{e}}phane Demri, Alain Finkel, Valentin Goranko, and Govert van Drimmelen.
\newblock Model-checking {CTL*} over flat {P}resburger counter systems.
\newblock {\em Journal of Applied Non-Classical Logics}, 20(4):313--344, 2010.
\newblock \href {http://dx.doi.org/10.3166/jancl.20.313-344}
  {\path{doi:10.3166/jancl.20.313-344}}.

\bibitem{DBLP:journals/tocl/DemriL09}
St{\'{e}}phane Demri and Ranko Lazic.
\newblock {LTL} with the freeze quantifier and register automata.
\newblock {\em {ACM} Trans. Comput. Log.}, 10(3):16:1--16:30, 2009.
\newblock \href {http://dx.doi.org/10.1145/1507244.1507246}
  {\path{doi:10.1145/1507244.1507246}}.

\bibitem{DBLP:conf/popl/EmersonH83}
E.~Allen Emerson and Joseph~Y. Halpern.
\newblock "{S}ometimes" and "not never" revisited: On branching versus linear
  time.
\newblock In John~R. Wright, Larry Landweber, Alan~J. Demers, and Tim
  Teitelbaum, editors, {\em Conference Record of the Tenth Annual {ACM}
  Symposium on Principles of Programming Languages, Austin, Texas, USA, January
  1983}, pages 127--140. {ACM} Press, 1983.
\newblock \href {http://dx.doi.org/10.1145/567067.567081}
  {\path{doi:10.1145/567067.567081}}.

\bibitem{DBLP:conf/fsttcs/FinkelL02}
Alain Finkel and J{\'{e}}r{\^{o}}me Leroux.
\newblock How to compose {P}resburger-accelerations: Applications to broadcast
  protocols.
\newblock In Manindra Agrawal and Anil Seth, editors, {\em {FST} {TCS} 2002:
  Foundations of Software Technology and Theoretical Computer Science, 22nd
  Conference Kanpur, India, December 12-14, 2002, Proceedings}, volume 2556 of
  {\em Lecture Notes in Computer Science}, pages 145--156. Springer, 2002.
\newblock \href {http://dx.doi.org/10.1007/3-540-36206-1\_14}
  {\path{doi:10.1007/3-540-36206-1\_14}}.

\bibitem{DBLP:conf/concur/HoenickeMO10}
Jochen Hoenicke, Roland Meyer, and Ernst{-}R{\"{u}}diger Olderog.
\newblock Kleene, rabin, and scott are available.
\newblock In Paul Gastin and Fran{\c{c}}ois Laroussinie, editors, {\em {CONCUR}
  2010 - Concurrency Theory, 21th International Conference, {CONCUR} 2010,
  Paris, France, August 31-September 3, 2010. Proceedings}, volume 6269 of {\em
  Lecture Notes in Computer Science}, pages 462--477. Springer, 2010.
\newblock \href {http://dx.doi.org/10.1007/978-3-642-15375-4\_32}
  {\path{doi:10.1007/978-3-642-15375-4\_32}}.

\bibitem{DBLP:conf/concur/KuhtzF11}
Lars Kuhtz and Bernd Finkbeiner.
\newblock Weak {K}ripke structures and {LTL}.
\newblock In Joost{-}Pieter Katoen and Barbara K{\"{o}}nig, editors, {\em
  {CONCUR} 2011 - Concurrency Theory - 22nd International Conference, {CONCUR}
  2011, Aachen, Germany, September 6-9, 2011. Proceedings}, volume 6901 of {\em
  Lecture Notes in Computer Science}, pages 419--433. Springer, 2011.
\newblock \href {http://dx.doi.org/10.1007/978-3-642-23217-6\_28}
  {\path{doi:10.1007/978-3-642-23217-6\_28}}.

\bibitem{DBLP:conf/time/LaroussinieMP10}
Fran{\c{c}}ois Laroussinie, Antoine Meyer, and Eudes Petonnet.
\newblock Counting {LTL}.
\newblock In Nicolas Markey and Jef Wijsen, editors, {\em {TIME} 2010 - 17th
  International Symposium on Temporal Representation and Reasoning, Paris,
  France, 6-8 September 2010}, pages 51--58. {IEEE} Computer Society, 2010.
\newblock \href {http://dx.doi.org/10.1109/TIME.2010.20}
  {\path{doi:10.1109/TIME.2010.20}}.

\bibitem{DBLP:journals/corr/abs-1211-4651}
Fran{\c{c}}ois Laroussinie, Antoine Meyer, and Eudes Petonnet.
\newblock Counting {CTL}.
\newblock {\em Logical Methods in Computer Science}, 9(1), 2012.
\newblock \href {http://dx.doi.org/10.2168/LMCS-9(1:3)2013}
  {\path{doi:10.2168/LMCS-9(1:3)2013}}.

\bibitem{DBLP:conf/atva/LerouxS05}
J{\'{e}}r{\^{o}}me Leroux and Gr{\'{e}}goire Sutre.
\newblock Flat counter automata almost everywhere!
\newblock In Doron~A. Peled and Yih{-}Kuen Tsay, editors, {\em Automated
  Technology for Verification and Analysis, Third International Symposium,
  {ATVA} 2005, Taipei, Taiwan, October 4-7, 2005, Proceedings}, volume 3707 of
  {\em Lecture Notes in Computer Science}, pages 489--503. Springer, 2005.
\newblock \href {http://dx.doi.org/10.1007/11562948\_36}
  {\path{doi:10.1007/11562948\_36}}.

\bibitem{minsky-computation-67}
M.~Minsky.
\newblock {\em Computation, Finite and Infinite Machines}.
\newblock Prentice Hall, 1967.

\bibitem{DBLP:conf/focs/Pnueli77}
Amir Pnueli.
\newblock The temporal logic of programs.
\newblock In {\em 18th Annual Symposium on Foundations of Computer Science,
  Providence, Rhode Island, USA, 31 October - 1 November 1977}, pages 46--57.
  {IEEE} Computer Society, 1977.
\newblock \href {http://dx.doi.org/10.1109/SFCS.1977.32}
  {\path{doi:10.1109/SFCS.1977.32}}.

\bibitem{Presburger29}
M.~Presburger.
\newblock {\"U}ber die {V}ollst{\"a}ndigkeit eines gewissen {S}ystems der
  {A}rithmetik ganzer {Z}ahlen, in welchem die {A}ddition als einzige
  {O}peration hervortritt.
\newblock In {\em Comptes Rendus du premier congr{\`e}s de math{\'e}maticiens
  des Pays Slaves, Warszawa}, pages 92--101, 1929.

\bibitem{DBLP:conf/pldi/Pugh94}
William Pugh.
\newblock Counting solutions to {P}resburger formulas: How and why.
\newblock In Vivek Sarkar, Barbara~G. Ryder, and Mary~Lou Soffa, editors, {\em
  Proceedings of the {ACM} SIGPLAN'94 Conference on Programming Language Design
  and Implementation (PLDI), Orlando, Florida, USA, June 20-24, 1994}, pages
  121--134. {ACM}, 1994.
\newblock \href {http://dx.doi.org/10.1145/178243.178254}
  {\path{doi:10.1145/178243.178254}}.

\bibitem{DBLP:journals/tocl/Schweikardt05}
Nicole Schweikardt.
\newblock Arithmetic, first-order logic, and counting quantifiers.
\newblock {\em {ACM} Trans. Comput. Log.}, 6(3):634--671, 2005.
\newblock \href {http://dx.doi.org/10.1145/1071596.1071602}
  {\path{doi:10.1145/1071596.1071602}}.

\end{thebibliography}

\newpage
\appendix

\section{Consistency Implies Correctness}
\label{apx:correctness}

This section is dedicated to proving \cref{lem:fltl-soundness}.

\fltlsoundness*

Recall that $\mathcal{K}=(S,s_I,E,\lambda)$ is a Kripke structure and $\Phi$ an \fLTL formula.
Let us first formally define the notion of correctness.
\begin{Definition}
  A location $\ell$ of $\mathcal{P}$ is \emph{correct} wrt. a formula $\xi$ if and only if
  \[
    \forall\sigma\in\Runs(\mathcal{P}): \forall i\in\mathbb{N}: \sigma(i)=\ell \Rightarrow (\xi\in\labels_\mathcal{P}(\ell) \Leftrightarrow (\sigma,i)\models\xi).
  \]
  An APS $\mathcal{P}$ is correct wrt.\ $\xi$ if that is the case for all locations of $\mathcal{P}$.
\end{Definition}
Notice that $\ell$ can only be consistent if $\xi$ holds at all positions where $\ell$ occurs or at none of them.

If $\mathcal{K}$ now contains an APS $\mathcal{P}$ that is correct wrt.\ $\Phi$ and that path schema contains a run $\sigma\in\Runs(\mathcal{P})$ then the correct labelling of the initial location 0 by $\Phi$ implies that $(\sigma,0)\models\Phi$ and thus $\mathcal{K}\models\Phi$.
Therefore, \cref{lem:fltl-soundness} is implied by the following, that we prove in the remainder of this section.

\begin{Lemma}\label{lem:fltl-correctness}
  If an APS $\mathcal{P}$ consistent wrt.\ $\Phi$ then it is correct wrt.\ $\Phi$.
\end{Lemma}

Let in the following $\mathcal{P}=(P_0,P_1,…,P_m)$ be an APS fixed and consistent wrt. $\Phi$.
Let further $\sigma\in\Runs(\mathcal{P})$ be any run of $\mathcal{P}$.
We use an induction over the structure of $\Phi$ to show that for all locations $\ell$ of $\mathcal{P}$ if $\Phi\in\labels_\mathcal{P}(\ell)$ then $\Phi$ holds at every occurrence of $\ell$ on $\sigma$ and if $\Phi\not\in\labels_\mathcal{P}(\ell)$ then $\Phi$ does not hold at any position where $\ell$ occurs on $\sigma$.

For easier reading, we use some abbreviations in the following.
Let $L_i=\labels_\mathcal{P}(\sigma(i))$ be the labelling of the location at position $i$ on $\sigma$ for $i\in\mathbb{N}$.
We also denote the set of occurrences of a location $\ell$ on $\sigma$ by $\sigma^{-1}(\ell) =\lbrace i\in\mathbb{N}\mid\sigma(i)=\ell\rbrace$.

\subsection{Propositions, Boolean Combinations and Temporal Next}

Let $i\in\mathbb{N}$ be any position on $\sigma$ and $\ell_i=\sigma(i)$ be the corresponding location in $\mathcal{P}$ with labelling $L_i$.
Consider the following cases for the structure of $\Phi$, the first being the induction base case.

\begin{description}
  \item[($\Phi=p\in AP$)] By consistency $p\in L_i \Leftrightarrow p\in\lambda(\st_\mathcal{P}(\ell_i))$ and by semantics $p\in\lambda(\st_\mathcal{P}(\ell_i)) \Leftrightarrow (\sigma,i)\models p$.
  \item[($\Phi=\neg\varphi$)] \[\neg\varphi\in L_i \stackrel{\text{consist.}}{\Leftrightarrow} \varphi\not\in L_i
                       \stackrel{\text{induct.}}{\Leftrightarrow} (\sigma,i)\not\models\varphi
                       \stackrel{\text{semant.}}{\Leftrightarrow} (\sigma,i)\models\neg\varphi\]
  \item[($\Phi=\varphi\land\psi$)] \[
    \varphi\land\psi\in L_i \stackrel{\text{consist.}}{\Leftrightarrow} \varphi,\psi\in L_i
          \stackrel{\text{induct.}}{\Leftrightarrow} (\sigma,i)\models\varphi \text{ and } (\sigma,i)\models\psi
          \stackrel{\text{semant.}}{\Leftrightarrow} (\sigma,i)\models\varphi\land\psi
  \]

  \item[($\Phi=\X\varphi$)] By the definition of a run we have $\sigma(i+1)\in\succpos_\mathcal{P}(\sigma(i))$ and thus
  \[
    \X\varphi\in L_i \stackrel{\text{consist.}}{\Leftrightarrow} \varphi\in L_{i+1}
          \stackrel{\text{induct.}}{\Leftrightarrow} (\sigma,i+1)\models\varphi
          \stackrel{\text{semant.}}{\Leftrightarrow} (\sigma,i)\models\X\varphi.
  \]
\end{description}

\subsection{Temporal Until}

Assume finally $\Phi=\varphi\U^{\frac{x}{y}}\psi$.
By consistency and induction $\mathcal{P}$ is correct wrt.\ $\varphi$ and $\psi$, thus $(\sigma,i)\models\varphi \Leftrightarrow \varphi\in L_i$ and $(\sigma,i)\models\psi \Leftrightarrow \psi\in L_i$ for all $i\in\mathbb{N}$.
Let $k=\comp_\mathcal{P}(\ell_i)$. Hence $P_k$ is the component that $\ell_i=\sigma(i)$ belongs to.
Further, for finite augmented paths $a_0…a_n$ let
\[
  \bal(a_0…a_n) = y\cdot|\lbrace j\in[0,n]\mid\varphi\in\labels(a_j)\rbrace| - x\cdot(n+1)
\]
denote the \emph{balance} between those positions that are labelled by $\varphi$ and those that are not, weighted according to the ratio required by $\Phi$.
That is, as discussed earlier, a “good” position contributes a reward of $y-x$ while a “bad” position causes a fee of $-x$.
Then, $\bal(a_0…a_n)\ge0$ is equivalent to the ratio condition $|\lbrace j\in[0,n]\mid\varphi\in\labels(a_j)\rbrace|\ge\frac{x}{y}\cdot(n+1)$ specified by $\Phi$ but allows us to reason on discrete integer numbers.
For convenience we apply this notation likewise for sequences
$v=v_1\ldots v_k$ of locations of $\mathcal{P}$ and write
$\bal(v):=\bal(\mathcal{P}(v1) \ldots \mathcal{P}(v_k))$.

We will now treat the different case of the notion of consistency.

\noindent\textbf{Case 3a}. Assume  $\Phi,\psi\in L_i$ then, by induction, $(\sigma,i)\models\psi$ which implies that $(\sigma,i)\models\Phi$.

\noindent\textbf{Case 3b}. $\ell_i$ is part of the final loop $P_m$ of
$\mathcal{P}$. Hence there is a smallest position $j>i$ where $\sigma(j)$ is part of the final loop $P$ of $\mathcal{P}$ and $\psi\in L_j$ (and thus holds there).
Hence, also $\psi\in L_{j+n|P|}$ for every number $n>0$.
Since $P$ is good for $\Phi$ the balance $\bal(P)>0$ is positive and thus
\[
  \bal(\ell_i…\ell_j)<\bal(\ell_i…\ell_{j+|P|})<…<\bal(\ell_i…\ell_{j+n|P|}).
\]
For sufficiently large $n$ (i.e., sufficiently many iterations of $P$) we obtain necessarily a non-negative balance $\bal(\ell_i…\ell_{j+n|P|})\ge0$ on the corresponding subpath of $\sigma$ and thus $(\sigma,i)\models\Phi$.

\noindent \textbf{Case 3c}. This is the case where we have a counter tracking the balance.
For this case to apply, $\ell_i$ must be part of a row and thus $i$ is the only position where $\ell_i$ occurs.
The condition requires that there is a counter $c$ of $\mathcal{P}$ that tracks the balance wrt. $\Phi$, starting at the occurrence of $\ell_i$.
Since $\sigma$ is a run, there is a corresponding sequence of valuations $\theta_0\theta_1…$ and for all $j\ge i$
\[
  \theta_j(c) = \bal(\sigma(i)\sigma(i+1)…\sigma(j-1)).
\]
If $\Phi\in\labels_\mathcal{P}(\sigma(i))$, the definition provides that there is a location $\ell'>\ell_i$ such that $\psi \in \labels_\mathcal{P}(\ell')$ and guarded by $(c\ge0)\in\guards_\mathcal{P}(\ell')$.
Thus, there is a position $j>i$ on $\sigma$ such that $\sigma(j) = \ell'$ and $\bal(\sigma(i)…\sigma(j-1)) = \theta_j \ge0$.
It follows that $(\sigma,i)\models\Phi$.
Similarly, if $\Phi\not\in\labels_\mathcal{P}(\sigma(i))$, then there is no position $j$ on $\sigma$ where $\psi$ holds and the balance between $i$ and $j$  is non-negative.
This is guaranteed because every such position carries a location $\ell'>\ell_i$ (since $\ell_i$ is not on a loop) and either $\psi\not\in\labels_\mathcal{P}(\ell')$ or $(c<0)\in\guards_\mathcal{P}(\ell')$ and thus $\bal(\sigma(i)…\sigma(j-1))=\theta_j(c)<0$.

\textbf{Case 3d}. It remains to consider the cases requiring a periodic sequence.
Let us start by establishing a lemma that provides a convenient argument for correctness and motivates the periodicity requirement imposed by the definition.

Recall that $\mathcal{P}=(P_0,P_1,…,P_m)$ and thus $\sigma$ has the form
\[
  \sigma = v_0^{n_0}v_1^{n_1}…v_{m-1}^{n_{m-1}}v_m^\omega
\]
where $v_j$, for $0\le j\le m$, is the sequence of locations corresponding to the $j$-th component, i.e.\ $\mathcal{P}[v_j]=P_j$, $n_j=1$ if $P_j$ is a row and $n_j\ge1$ if $P_j$ is a loop.

\begin{Lemma} \label{lem:periodic}
  Let $(P_j,P_{j+1},…,P_{\jhat})$ be a $\lbrace\varphi,\psi\rbrace$-periodic sequence of components of $\mathcal{P}$ with $0\le j<\jhat<m$ and $\sigma=uvw$ for $v=v_j^{n_j}v_{j+1}^{n_{j+1}}…v_{\jhat}^{n_{\jhat}}$.
  Let further $|u|<i_1\le i_2$ be positions on $\sigma$ and $n\in\mathbb{N}$ such that $i_2=i_1+n|P_j|<|uv|$.
  \begin{enumerate}
    \item If $P_j$ is good or neutral for $\Phi$ then $(\sigma,i_2)\models\Phi \Rightarrow (\sigma,i_1)\models\Phi$.
    \item If $P_j$ is bad or neutral for $\Phi$ and $i_2<|uv|-|P_j|$ then $(\sigma,i_1)\models\Phi \Rightarrow (\sigma,i_2)\models\Phi$.
  \end{enumerate}
\end{Lemma}

\begin{proof}
\begin{description}
  \item[1.] Assuming $(\sigma,i_2)\models\Phi$ there is a position $i_3 \ge i_2$ such that $(\sigma,i_3)\models\psi$ and $\bal(\sigma(i_2)…\sigma(i_3-1))\ge0$.
    Due to $\varphi$-periodicity we have also
    \begin{align*}
      \bal(\sigma(i_1)…\sigma(i_2-1)…\sigma(i_3)) &= \bal(\sigma(i_1)…\sigma(i_1+n|P_j|-1)…\sigma(i_3-1))) \\
      & = n\cdot\bal(P_j) + \bal(\sigma(i_2)…\sigma(i_3-1))\ge0
    \end{align*}
  \item[2.] Assuming $(\sigma,i_1)\models\Phi$ there is a position $i_3 \ge i_1$ such that $(\sigma,i_3)\models\psi$ and $\bal(\sigma(i_1)…\sigma(i_3-1))\ge0$.
  If $i_3\ge i_2$ we have
  \begin{align*}
    0 & \le \bal(\sigma(i_1)…\sigma(i_1+n|P_j|-1)\sigma(i_2)…\sigma(i_3-1)) \\
      & = n\bal(P_j) + \bal(\sigma(i_2)…\sigma(i_3-1))\\
      & \le \bal(\sigma(i_2)…\sigma(i_3-1))
  \end{align*}
  due to $\varphi$-periodicity and $\bal(P_j)\le0$.

  If $i_3<i_2$ we can assume w.l.o.g.\ that $i_3<i_1+|P_j|$ because otherwise we can also choose $i_3-|P_j|$ as witness instead of $i_3$:
  $\psi\in\labels_\mathcal{P}(\sigma(i_3-|P_j|))$ due to $\psi$-periodicity and since $\bal(P_j)\le0$ we would have
  \begin{align*}
    0 &\le \bal(\sigma(i_1)\sigma(i_1+1)…\sigma(i_3-|P_j|-1)…\sigma(i_3-1))\\
      &= \bal(\sigma(i_1)\sigma(i_1+1)…\sigma(i_3-|P_j|-1)) + \bal(P_j)\\
      &\le \bal(\sigma(i_1)\sigma(i_1+1)…\sigma(i_3-|P_j|-1)).
  \end{align*}
  Repeating this argument eventually provides a witness $i_3<i_1+|P_j|$.

  Then,
  \begin{align*}
    0 & \le \bal(\sigma(i_1)…\sigma(i_3-1)) \\
      & = \bal(\sigma(i_1+n|P_j|)…\sigma(i_3 - 1 + n|P_j|)) \\
      & = \bal(\sigma(i_2)…\sigma(i_3 - 1 + n|P_j|))
  \end{align*}
  because position $i_3+n|P_j| < i_1+(n+1)|P_j|= i_2+|P_j| < |uv|$ on $\sigma$ still carries a location from the periodic part $(P_{k'},…,P_{\jhat})$ of $\mathcal{P}$.
  For the same reason we have $\psi\in\labels_\mathcal{P}(\sigma(i_3+n|P_j|))$ and thus $\Phi$ holds at position $i_2$.
\end{description}
\end{proof}

Based on \cref{lem:periodic} correctness can easily be
established. The definition demands a component $P_{k'}$ where each
location is consistent and as shown earlier we can assume that it is
thus correct not only wrt.\ $\varphi$ and $\psi$ but also wrt.\
$\Phi$. We do then the following case analysis:
\begin{description}
  \item[$k=m$:] Considering the final loop $P_m$ we have a preceding correct component $P_k$ for $k'<m$.
    Periodicity wrt.\ $\Phi$ provides that for some $n\in\mathbb{N}$ the position $i'=i - n|P_k|$ on $\sigma$ carries a location $\sigma(i')\in\loc_\mathcal{P}(k')$ from $P_{k'}$ and $\Phi\in\labels_\mathcal{P}(\sigma(i')) \Leftrightarrow \Phi\in\labels_\mathcal{P}(\sigma(i))$.
    Due to periodicity (and correctness) wrt.\ $\varphi$ and $\psi$, the formula $\Phi$ cannot distinguish any of the positions $i'+n'|P_k|$, i.e., $(\sigma,i')\models\Phi$ iff $(\sigma,i'+n'|P_k|)\models\Phi$ for any $n'\in\mathbb{N}$ since the infinite suffix $\sigma(i')\sigma(i'+1)…$ is equivalent to every suffix $\sigma(i'+n'|P_k|)\sigma(i'+n'|P_k|+1)…$ regarding the positions where $\varphi$ and $\psi$ hold.
    Hence,
    \[
      \Phi\in\labels_\mathcal{P}(\sigma(i)) \quad \Leftrightarrow \quad \Phi\in\labels_\mathcal{P}(\sigma(i')) \quad \Leftrightarrow \quad (\sigma,i')\models\Phi \quad \Leftrightarrow \quad (\sigma,i)\models\Phi.
    \]

  \item[$P_k$ is good or neutral for $\Phi$ and $\Phi\in\labels_\mathcal{P}(\sigma(i))$:]
    $k'>k$ and there is $i'=i + n |P_k|$ for some (unique) $n$ such that $\sigma(i')\in\loc_\mathcal{P}(k')$.
    Due to $\Phi$-periodicity, we have that $\Phi\in\labels_\mathcal{P}(\sigma(i'))$ and thus by correctness of that labelling and \cref{lem:periodic} we have $(\sigma,i)\models\Phi$.
  \item[$P_k$ is good or neutral for $\Phi$ and $\Phi\not\in\labels_\mathcal{P}(\sigma(i))$:]
    $k'<k$ and there is $i'=i - n |P_k|$ for the unique $n$ such that $\sigma(i')\in\loc_\mathcal{P}(k')$.
      Due to $\Phi$-periodicity, we have that $\Phi\not\in\labels_\mathcal{P}(\sigma(i'))$ and thus $(\sigma,i')\not\models\Phi$ which implies by \cref{lem:periodic} that $(\sigma,i)\not\models\Phi$.
  \item[$P_k$ is bad for $\Phi$ and $\Phi\in\labels_\mathcal{P}(\sigma(i))$:]
    $k'<k$ and there is $i'=i - n |P_k|$ for the unique $n$ such that $\sigma(i')\in\loc_\mathcal{P}(k')$.
    We have that also $\Phi\in\labels_\mathcal{P}(\sigma(i'))$.
    Since $i$ is a position in an iteration of $P_k$ at least one iteration of $P_{k+1}$ follows this position on $\sigma$, which still belongs to the periodic sequence.
    Therefore we can apply \cref{lem:periodic} and conclude from $(\sigma,i')\models\Phi$ that $(\sigma,i)\models\Phi$.
  \item[$P_k$ is bad for $\Phi$ and $\Phi\not\in\labels_\mathcal{P}(\sigma(i))$:] $k'>k$ and there is $i'=i + n |P_k|$ for some (unique) $n$ such that $\sigma(i')\in\loc_\mathcal{P}(k')$.
    Again, periodicity and the guaranteed additional iteration of $P_{k'+1}$ after $i'$ on $\sigma$ allows for applying \cref{lem:periodic} and to conclude from $(\sigma,i')\not\models\Phi$ that $(\sigma,i)\not\models\Phi$.

\end{description}
 
\section{Constructing Path Schemas from Satisfying Runs}
\label{apx:completeness}

This section is dedicated to proving \cref{lem:fltl-completeness}.
We will use the notation $\bal$ and $\sigma^{-1}$ for a run $\sigma$ of an APS introduced in \cref{apx:correctness}.
Recall that $\mathcal{K}$ is a flat Kripke structure that admits a run $\rho\in\Runs(\mathcal{K})$ and $\Phi$ is an \fLTL formula.
Assume for this section that $\rho\models\Phi$.
From $\rho$ we construct a non-empty APS $\mathcal{P}_\Phi$ that is consistent wrt.\ $\Phi$ and of which the first location $\mathcal{P}[0]$ is labelled by $\Phi$.
In fact, it admits a run $\sigma$ representing $\rho$, i.e.\ such that $\st_\mathcal{P}(\sigma)=\rho$.
The construction provides an exponential bound on the size of $\mathcal{P}_\rho$ and thereby proves the lemma.

\fltlcompleteness*

Every run $\rho\in\Runs(\mathcal{K})$ can be represented by a small path schema in $\mathcal{K}$, labelled only by propositions.
The labelling is then extended stepwise to include larger and larger subformulae of $\Phi$ until all subformulae and finally $\Phi$ itself are consistently annotated.
Every step needs to ensure that the new annotation is consistent, which may require the modification of the structure, namely unfolding and duplication of loops.
Thus, we also need to argue that after each modification there is still a valid run that represents $\rho$ and the obtained schema grew only linearly in size.

We use induction over the structure of $\Phi$ starting by its base case of $\Phi\in AP$ being atomic.

\subparagraph{Base case.}

Since $\mathcal{K}$ is flat, any subpath $\rho(i)\rho(i+1)…\rho(i')…\rho(i'')$ of $\rho$ where a state $\rho(i)=\rho(i')=\rho(i'')$ occurs more than twice is equal to $(\rho(i)(i+1)…\rho(i'-1))^2\rho(i'')$.
Hence, there are simple subpaths $u_0,…,u_m\in S^+$ of $\rho$ and positive numbers of iterations $n_0,…,n_{m-1}\in\mathbb{N}$ such that
\[
  \rho=u_0^{n_0}u_1^{n_1}…u_{m-1}^{n_{m-1}}u_m^\omega
\]
and $|u_0u_1…u_m|\le2|S|$.
They naturally induce the augmented path schema $\mathcal{P}_\Phi=(P_0,…,P_m)$ where the augmented paths $P_k$ correspond directly to the paths $u_k$.
Formally, for $u_k=s_0…s_n$ we let
\[
  P_k=(s_0,\lambda(s_0),\emptyset,\mathbf{0},t_0)…(s_k,\lambda(s_k),\emptyset,\mathbf{0},t_n)
\]
where, for $0\le i\le n$, the type of each augmented state is $t_i=\mathtt{R}$ if it is iterated $n_i=1$ times and $t_i=\mathtt{L}$ if it is iterated $n_i>1$ times on $\rho$.
By construction $\mathcal{P}_\Phi$ is consistent wrt.\ any proposition from $AP$ and we have a run $\sigma\in\Runs(\mathcal{P})$ such that $\st_\mathcal{P}(\sigma)=\rho$.

The path schema does not use any counter so we can consider the set of counters $C$ to be empty.
During the following constructions we may introduce new counters.
Technically that means we would have to adjust all augmented states, simply because the signature of updates changes.
For convenience, we therefore implicitly extend update functions and assigning zero to counter names if not explicitly stated otherwise.

Now, building on this base case, we show how to construct $\mathcal{P}_\Phi$ assuming by induction that there is an APS $\mathcal{P}$ that contains a run $\sigma\in\Runs(\mathcal{P})$ with $\st_{\mathcal{P}}(\sigma)=\rho$ and is consistent with respect to all strict subformulae of $\Phi$.

\subparagraph{Boolean combinations.} If $\Phi$ is a boolean combination then every augmented state $a=(s,L,G,\mathbf{u},t)$ in $\mathcal{P}$ is easily adjusted to obey \cref{def:consistency}.
For $\Phi=\neg\varphi$ we add $\Phi$ to $L$ if and only if $\varphi\not\in L$.
For $\Phi=\varphi\land\psi$ we add $\Phi$ to $L$ if and only if $\varphi,\psi\in L$.
These changes do not modify the set of runs and so $\sigma$ remains a run in the obtained structure $\mathcal{P}_\Phi$.

\subsection{Temporal Next}

For $\Phi=\X\varphi$ the labelling at some location $\ell$ is extended according to the labelling of its successors.
If $\varphi\in\labels_\mathcal{P}(\ell')$ for all $\ell'\in\succpos_\mathcal{P}(\ell)$ then we modify $\labels_\mathcal{P}(\ell)$ such that it contains $\Phi$ and if $\varphi\not\in\labels_\mathcal{P}(\ell')$ for all successors $\ell'$ then the labelling remains untouched, not including $\Phi$.

The labellings, however, may disagree upon $\varphi$ if $\ell$ is the last location in a loop $P_k$ of $\mathcal{P}$.
In that case the loop $P_k$ needs to be removed or unfolded, in order to make $\mathcal{P}$ consistent.
For augmented states let $\mathtt{row}((s,L,G,\mathbf{u},t)):=(s,L,G,\mathbf{u},\mathtt{R})$ denote the same state but with type $\mathtt{R}$ and for sequences let $\mathtt{row}(a_0…a_n):=\mathtt{row}(a_0)…\mathtt{row}(a_n)$.

Now, if the run $\sigma$ takes $P_k$ only once it can be cut by replacing it with $P'_k=\mathtt{row}(P_k)$.
This eliminates runs that take $P_k$ more than once but $\sigma$ remains.
If otherwise $\sigma$ takes $P_k$ at least twice, the loop can be unfolded by inserting $P'_k$ between $P_k$ and $P_{k+1}$, i.e.\ letting
\[
  \mathcal{P}'=(P_0,…,P_k,P_k',P_{k+1},…,P_m).
\]
The run $\sigma$ representing $\rho$ persists, up to adjusting it according to the new shifted indices due to the insertion of $P'_k$.
Formally, $\sigma$ has the form
\[
  \sigma = v_0^{n_0}…v_{m-1}^{n_{m-1}}v_m^\omega.
\]
where $v_i= \ell_i \, (\ell_i+1) \, … \, (\ell_i+|P_i|-1)$ for $0\le i\le m$ and $\ell_i=|P_0…P_{i-1}|$.
Hence, there is a run $\sigma'\in\Runs(\mathcal{P}')$ with
\[
  \sigma'(i) = \begin{cases}
    \sigma(i) & \text{if $i<|v_0^{n_0}…v_k^{n_k-1}|$}\\
    \sigma(i)+|P_k| & \text{otherwise.}
  \end{cases}
\]
and thus $\st_{\mathcal{P}'}(\sigma')=\st_\mathcal{P}(\sigma)=\rho$.

Importantly, cutting or unfolding any loop, even any number of times, in $\mathcal{P}$ preserves consistency.

\begin{Lemma}
  Let $\mathcal{P}=(P_0,…,P_k,…,P_m)$ be an APS, $P_k$ a loop in $\mathcal{P}$ that is consistent with respect to an \fLTL formula $\xi$ and  $P'_k=\mathtt{row}(P_k)$ an unfolding.
  The component $P'_k$ and all components $P_h$ that are consistent with respect to $\xi$ in $\mathcal{P}$ are also consistent with respect to $\xi$ in all of the following APS:
  \begin{itemize}
    \item $\mathcal{P}_\textsf{cut}=(P_0,…,P_{k-1},P'_k,P_{k+1},…,P_m)$
    \item $\mathcal{P}_\textsf{left}=(P_0,…,P_{k-1},P'_k,P_k,…,P_m)$
    \item $\mathcal{P}_\textsf{right}=(P_0,…,P_k,P'_k,P_{k+1},…,P_m)$
    \item $\mathcal{P}_\textsf{dupl}=(P_0,…,P_k,P_k,P_{k+1},…,P_m)$
\end{itemize}
\end{Lemma}

\begin{proof}[Sketch of proof]
The cases for propositions and Boolean combinations are straightforward.
Considering formulae $\xi=\X\varphi$, an easy case analysis reveals that all combinations of locations and their successors correspond to a similar combination that occurs in $\mathcal{P}$.
Considering until formulae, a location on $P_k$ in $\mathcal{P}$ can only be consistent because of a consistent component $P_{k'}$ (as detailed in Case {\bf 3d}).
This condition applies equally if $P_k$ is made a row.
If copies (loops or rows) of $P_k$ are inserted, the share the same labelling by $\xi$ and its subformulae and therefore smoothly integrate in any relevant repeating sequence.
For example if $\xi\in\labels_\mathcal{P}(\ell)$ for some location on $P_k$ and $P_k$ is bad for $\xi$, then there is a repeating sequence starting in some consistent component and ending in $P_k$.
A copy of $P_k$ to the right extends this sequence which then provides the reason for the copy to be also consistent and a copy to the left does not break the sequence because it is labelled the same as $P_k$.
The same applies for $\xi\not\in\labels_\mathcal{P}(\ell)$ and similarly if $P_k$ is good or neutral for $\xi$.
\end{proof}

\subsection{Until}

Assume now that $\Phi=\varphi\U^{\frac{x}{y}}\psi$ is an until formula.
In order to construct $\mathcal{P}_\Phi$ given $\mathcal{P}$ and $\sigma$ we iterate through the components of $\mathcal{P}$, beginning at the last and transforming them one by one until the first.
The invariant is that the number of components that are yet to be considered becomes smaller by one in each step (although the overall number of components may increase) and that there is always a run representing $\rho$.
The following lemma formalises one such step.

\begin{Lemma}\label{lem:until-component}
  Let $\mathcal{P}=(P_0,…,P_m)$ be an augmented path schema, $k\in[0,m]$ and $\ell=|P_0…P_{k-1}|$ such that
\begin{itemize}
  \item $\mathcal{P}$ is consistent wrt.\ $\varphi$ and $\psi$,
  \item every location $\ell'\in[\ell+|P_k|,|\mathcal{P}|-1]$ is consistent wrt.\ $\Phi$ and
  \item there is a run $\sigma\in\Runs(\mathcal{P})$ with $\st_\mathcal{P}(\sigma)=\rho$.
\end{itemize}

There is an augmented path schema $\mathcal{P}'=(P_0,…,P_{k-1},P'_k,…,P'_{m'})$ such that
  \begin{itemize}
    \item $\mathcal{P}'$ is consistent wrt.\ $\varphi$ and $\psi$,
    \item every location $\ell'\in[\ell,|\mathcal{P}'|-1]$ is consistent wrt.\ $\Phi$,
    \item there is a run $\sigma'\in\Runs(\mathcal{P}')$ with $\st_{\mathcal{P}'}(\sigma')=\rho$ and
    \item $|\mathcal{P}'|\le|\mathcal{P}| + 17y|\mathcal{K}|^3$.
  \end{itemize}
\end{Lemma}

For $0\le k\le m$ let $\ell_k=|P_0…P_{k-1}|$ be the first location in $\mathcal{P}$ corresponding to component $P_k$.

\begin{proof}
We proceed by a case analysis.

\subparagraph{Final loop.}

Assume that $k=m$, thus $P_k=P_m$ is the final loop in $\mathcal{P}$.
If Case {\bf 3b} of \cref{def:consistency} applies, all $P_m$ is to be entirely labelled by $\Phi$.
Otherwise, we consider the first iteration of $P_m$, starting at position $i_m:=\min \sigma^{-1}(\ell_m)$  (where $\ell_m$ is the first location of $P_m$) and have $P_m(j)$ labelled by $\Phi$ if and only if $(\sigma,i_m+j)\models\Phi$ for $0\le j<|P_m|$.\footnote{Notice that for proving the statement it is not necessary to be constructive. It suffices to observe that such a labelling exists.}
Hence, those states labelled by $\psi$ are labelled by $\Phi$ which is consistent.
If there are others states that we label by $\Phi$ and that are not labelled by $\psi$, we unfold $P_m$ twice and hence let $P'=(P_0,…,P_{m-1},P'_m,P'_{m+1},P'_{m+2})$ for  $P'_m = P'_{m+1} = \mathtt{row}(P_m)$ and $P'_{m+2} = P_m$.

The locations $\ell_m,…,\ell_m+|P_m|-1$, now associated with $P'_m$, can be made consistent (case {\bf 3c}).
For every location $\ell\in[\ell_m,\ell_m+|P_m|+1]$ we introduce a fresh counter $c$ that is updated on the locations succeeding $\ell$ as required by the definition.
If $\Phi\not\in\labels_{\mathcal{P}'}(\ell)$ we only need to add the guard $c<0$ to the states at those locations $\ell<\ell'<|\mathcal{P}'|$ that are labelled by $\psi$.
If $\Phi\in\labels_{\mathcal{P}'}(\ell)$ then because $\Phi$ holds at its first occurrence $j$ on $\sigma$.
In that case, if $\psi\not\in\labels_{\mathcal{P}'}(\ell)$, there must be a position $j'>j$ on $\sigma$ where $\psi$ holds and that is reached with positive balance.
Second, $P_m$ must be bad for $\Phi$, because otherwise the case above applied already, and thus we can assume without loss of generality that $j'<j+|P_m|$ and hence location $\ell'=\sigma(j')$ carries a state from $P'_m$ or $P''_m$.
They are both rows and therefore $\ell'$ can serve as witness location to be guarded by $(c\ge0)$.

The locations $\ell\in[\ell_m+|P_m|,\ell_m+2|P_m|-1]$ of $P''_m$ where $\Phi\not\in\labels_{\mathcal{P}'}(\ell)$ can also be made  consistent by adding a fresh counter that is updated an guarded as required.
Again, the case that $\Phi\in\labels_{\mathcal{P}'}(\ell)$ only occurs if $P_m$ is bad.
In that case $\ell$ is consistent already because $P''_m$ is part of the $\lbrace\varphi,\psi,\Phi\rbrace$-repeating sequence $(P'_m,P''_m,P_m)$ where $P'_m$ is  consistent (case {\bf 3d}).
The final loop $P_m$ is consistent for the same reason.

The size of the final loop is bounded by $|P_m|\le|\mathcal{K}|$ and at most two new copies of it are added to obtain $\mathcal{P}'$.

\subparagraph{Rows.}

Assume $k\in[0,m-1]$ and $P_k=a_0…a_n$ is a row in $\mathcal{P}$ starting at location $\ell_k=|P_0…P_{k-1}|$.
Let $h\in\mathbb{N}$ be the position on $\sigma$ with $\sigma(h)=\ell_k$.

We first adjust the labelling of $P_k$ such that $\Phi\in\labels(a_i)$ if and only if $(\sigma,h+i)\models\Phi$ for $i\in[0,n]$.
Now, with every location $\ell_i\in[\ell_k,\ell_k+n]$ that is not already consistent with respect to $\Phi$ (because of case {\bf 3.a} in \cref{def:consistency}) we proceed as follows.
A fresh counter $c$ is introduced and updated at all locations $\ell\in[\ell_i+1,|\mathcal{P}'|-1]$ to count the balance as required in the definition.
If $\Phi\not\in\labels(a_i)$ then the additional guard $c<0$ is added to the augmented states at those locations $\ell\in[\ell_i+1,|\mathcal{P}'|-1]$ that are labelled by $\Psi$.
This makes $\ell_i$ consistent and moreover, since $\Phi$ does not hold at position $h+i$ on $\sigma$, these constraint are not violated by the run.

If $\Phi\in\labels(a_i)$ (although $\psi\not\in\labels(a_i)$) there is a position $h'>h+i$ such that
$(\sigma,h')\models\psi$ and thus $\psi\in\labels_\mathcal{P}(\ell')$ for $\ell'=\sigma(h')$.
If $\ell'$ is on a row we add the constraint $c\ge0$ to the augmented state at $\ell'$.
If $\ell'$ is on a loop $P_{k'}$ of $\mathcal{P}$ but $\sigma$ takes it only once we can replace $P_{k'}$ by $\row(P_{k'})$ in $\mathcal{P}$ and then add the constraint.
In case $\sigma$ takes $P_{k'}$ at least twice either the last or first iteration of $P_{k'}$ can serve as a witness:
If $P_{k'}$ is good for $\Phi$ more iterations of it between $h$ and $h'$ can only improve the balance so that the ratio between $h$ and the last position $h''=\max \sigma^{-1}(\ell')$ where $\ell'$ is sufficient for $\Phi$.
Hence, we unfold $P_{k'}$ by adding a copy $\row(P_{k'})$ right after $P_k$ in $\mathcal{P}$.
Notice that we can assume that $k'<m$ because if $P_{k'}$ were the final loop and good for $\Phi$ the case above had already applied and we would not need to unfold the loop. Similarly, if $P_{k'}$ is bad (or neutral) for $\Phi$ then we let $h''=\min \sigma^{-1}(\ell')$ be the first position where $\ell'$ occurs.
Since $\bal(\sigma(h)…\sigma(h'))\le\bal(\sigma(h)…\sigma(h''))$ in this case $h''$ also can serve as witness and we unfold the loop by inserting $P'_{k'}$ immediately before $P_{k'}$ in $\mathcal{P}$.

As argued earlier, these transformations do not make any consistent location inconsistent with respect to any formula and there is still a run $\sigma'$ representing $\st_\mathcal{P}(\sigma)=\rho$.
However, the location $\ell''$ (at position $h''$ on $\sigma'$) is not part of a loop and can safely be guarded by $(c\ge0)$ while preserving the run.

During this procedure we introduce at most one unfolding of some loop for each position on $P_k$ and the size of $\mathcal{P}$ increases thus by at most $|\mathcal{K}|^2$ because $|\mathcal{K}|$ bounds the length of each loop.

\subparagraph{Non-final Loops.}

It remains to consider the case that $P_k$ is a non-final loop.
The run $\sigma$ has the form $\sigma=uv^nw$ where $v=\ell_k \, (\ell_k+1)\, … \, (\ell_k+|P_k|-1|)$ is the sequence of locations corresponding to $P_k$ in $\mathcal{P}$ and $n\in\mathbb{N}$ is maximal, that is $u$ and $w$ do not intersect with $v$.
We assume in the following that $n$ is not small as otherwise we may simply replace $P_k$ by $n$ copies of $\row(P_k)$ and proceed as above.
More precisely, let $\hat{n}=y\cdot|P_k|$ and assume that $n\ge\hat{n}+2$.
This constant $\hat{n}$ essentially bounds the effect of frequency variations within a single loop iteration.
Its specific choice will become apparent in the later construction.
For now it suffices to observe that if $n\le\hat{n}+1=y|P_k|+1\le y|\mathcal{K}|+1$ and we replace $P_k$ by $n$ unfoldings the size of $\mathcal{P}$ increases by $(n-1)\cdot|P_k|$.
Applying the procedure for rows above to each component may force us to unfold other loops.
As a (rough) estimate, we will have to introduce no more than one further unfolding of some loop for each new location originating from the unfoldings of $P_k$.
Hence, after making all $n$ copies of $P_k$ consistent the size of $\mathcal{P}$ did not grow by more than
\[
  (n-1)\cdot|P_k| + n\cdot|P_k|\cdot|\mathcal{K}| \le (y|\mathcal{K}|+1-1)\cdot|\mathcal{K}| + (y|\mathcal{K}|+1)\cdot|\mathcal{K}| \cdot|\mathcal{K}| \le 3y|\mathcal{K}|^3.
\]

Given that $n\ge\hat{n}+2$ we distinguish two situations of $\sigma$ determining a labelling for $P_k$.
Either, for all position $|u|\le i<|uv^{n-1}|$ we have $(\sigma,i)\models\Phi \Leftrightarrow (\sigma,i+|v|)\models\Phi$, meaning that the labelling of the augmented state $\mathcal{P}(\ell)$ at location $\ell$ on $v$ is unambiguously determined by $\sigma$ (we say that the loop is {\bf stable}), or there is a location on $v$ such that at some of its occurrence on $\sigma$ the formula $\Phi$ holds while at another it does not (in that case the loop is {\bf unstable}).
We consider first the former case and how it can be made consistent.
Afterwards we show that in the latter case it is possible to modify $\mathcal{P}$ such that the former case applies.

\subparagraph{Stable loops.}
If the pattern of positions where $\Phi$ holds is stable along the iterations of $P_k$ on $\sigma$ we apply it to the labelling of $P_k$.
That is, we adjust $P_k$ such that $\Phi\in\labels(P_k(i))$ if and only $(\sigma,|u|+i)\models\Phi$.
Likely, at least some of the locations $\ell\in[\ell_k,\ell_k+|P_k|-1]$ are still not consistent with respect to $\Phi$.
If $n\le4$ we replace $P_k$ in $\mathcal{P}$ by $n$ unfoldings $P'_k=\row(P_k)$ that can be made consistent as above.
Otherwise, let $R_1=R_2=R_3=R_4=P'_k$ and insert $(R_1,R_2)$ before and $R_3,R_4$ after $P_k$ in $\mathcal{P}$.
The two last unfoldings $R_3$ and $R_4$ can be made consistent as above.
For $R_1$ we proceed the same way except that if $P_k$ is to be unfolded again (for instance to find a location labelled with $\psi$) $R_2$ or $R_3$ are considered instead.
Now, $R_2$ and $P_k$ are also  consistent because the surrounding components $R_1,P_k,R_3,R_4$ cover every possible case.
Overall no more than $4$ additional copies of $P'_k$ are added and for the locations of at most three of them other loops needed to be unfolded giving a total of no more than
\[
 4\cdot|P_k| + 3\cdot|P_k|\cdot|\mathcal{K}|\le7|\mathcal{K}|^2
\]
new locations being added to $\mathcal{P}$.

\subparagraph{Unstable loops.}

In general, $\sigma$ does not uniquely determine whether the state at some location $\ell\in[\ell_k,\ell_k+|P_k|-1]$ in $\mathcal{P}$ is supposed to be labelled by $\Phi$ because that may vary between corresponding position on $\sigma$, that is, the iterations of $P_k$.
However, we observe that along any run the validity of $\Phi$ at some specific location can change at most once.
We have argued earlier that as soon as $\Phi$ holds somewhere, more iterations of a good loop inserted between the position in question and a witness position does not affect validity.
Similarly, introducing additional iterations of a bad loop do not change the fact that $\Phi$ does not hold at some specific position.

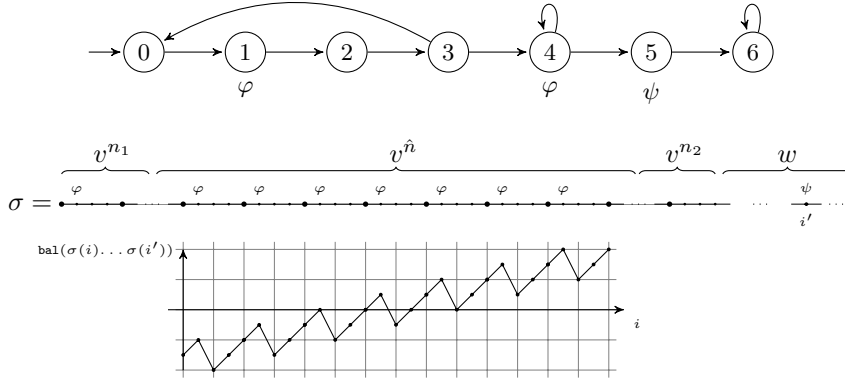
\begin{figure}
  
\begin{tikzpicture}

\node (ps) {
\begin{tikzpicture}[small nodes]
    \node[state, initial] (q0) {$0$};
    \node[state, right of = q0, label=below:$\varphi$] (q1) {$1$};
    \node[state, right of = q1] (q2) {$2$};
    \node[state, right of = q2] (q3) {$3$};
    \node[state, right of = q3, label=below:$\varphi$] (q4) {$4$};
    \node[state, right of = q4, label=below:$\psi$] (q5) {$5$};
    \node[state, right of = q5] (q6) {$6$};

    \path[->] (q0) edge (q1)
              (q1) edge (q2)
              (q2) edge (q3)
              (q3) edge (q4)
              (q4) edge (q5)
              (q5) edge (q6)

              (q3) edge[bend right] (q0)
              (q4) edge[loop above left] (q4)
              (q6) edge[loop above left] (q6);

\end{tikzpicture}
};

\path (ps.south) node[anchor = north] {
\begin{tikzpicture}[x = 2mm, y=2mm]
  \path[draw, very thin, gray] (-0.5,-4.5) grid[step=2] (28.5,4.5);
  \draw[->] (-1,0) -- (29,0) node[anchor = north west] {\tiny$i$};
  \draw[->] (0,-4) -- (0,4) node [anchor = east] {\tiny$\bal(\sigma(i)…\sigma(i'))$};

    \draw[very thin] (-8,7) node[anchor = east] {$\sigma = $}
    -- (-3,7);
  \draw[dotted] (-2.5,7) -- (-1.5,7);
  \draw[very thin] (-1,7) -- (29,7);
  \draw[dotted] (29.5 ,7) -- (30.5,7);
  \draw[very thin, fill] (31,7) -- (36,7);
      \draw[dotted] (37.5,7) -- (38.5,7);

  \draw[very thin, fill] (40,7) -- (41,7) circle (0.5pt) -- (42,7);
  \path (41,7) node [anchor = north] {\tiny$i'$};
  \path (41,7) node [anchor = south] {\tiny$\psi$};
  \draw[dotted] (42.5,7) -- (43.5,7);

    \path[fill] (-8,7)
    \foreach \x in {-8,-4, 0, 4, ..., 32}{
      -- (\x,7) circle (1pt)
    };

    \foreach \x in {-7, 1, 5, ..., 28}{
        \node[anchor = south] at (\x,7) {\tiny$\varphi$};
  };

    \path[fill] (-5,7)
      \foreach \x in {-7, -6, -5, 1, 2, ..., 28, 33, 34, 35}{
        -- (\x,7) circle (0.5pt)
      };

  \draw decorate[decoration=brace] {(-8,9) -- node[anchor=south, xshift=1mm, yshift=1] {$v^{n_1}$} (-2.25,9)};
  \draw decorate[decoration=brace] {(-1.75,9) -- node[anchor=south, xshift=1mm, yshift=1] {$v^{\hat{n}}$} (29.75,9)};
  \draw decorate[decoration=brace] {(30,9) -- node[anchor=south, xshift=1mm, yshift=1] {$v^{n_2}$} (35,9)};
  \begin{scope}
    \path [clip] (35.5,5) rectangle (43.5,11);
    \draw decorate[decoration=brace] {(35.5,9) -- node[anchor=south, yshift=1] {$w$} (44,9)};
  \end{scope}

\tikzmath{
for \n in {0,...,6}{
  \x = \n*4;
  {
  \draw[fill] (0+\x,0+\n-3) circle (0.5pt) --
              (1+\x,1+\n-3) circle (0.5pt) --
              (2+\x,-1+\n-3) circle (0.5pt) --
              (3+\x,0+\n-3) circle (0.5pt) --
              (4+\x,1+\n-3) circle (0.5pt);
  };
};
}

\end{tikzpicture}
};

\end{tikzpicture}
   \caption{Evolution of the balance representing the frequency constraint of $\Phi=\varphi\U^{\frac{1}{3}}\psi$ between positions $i$ and a fixed target position $i'$ along run $\sigma=v^nw$.
    The $n$ iterations of the loop $v=0 \, 1 \, 2 \, 3$ are separated into three parts as $n=n_1+\hat{n}+n_2$.
    The formula $\Phi$ holds during the iteration of $v$ whenever the balance is non-negative. Notice that there is a first position where the balance becomes non-negative and a last position where the balance is still negative.
    The validity of $\Phi$ only swaps in between.
    \label{fig:balance-graph}}
\end{figure}

It follows, for example, that if $\Phi$ does hold in the last iteration of a bad loop but not in the first, there is a unique iteration for each location on the loop where validity swaps.
The diagram presented in \cref{fig:balance-graph} shows an example of how the balance between a position $i$ on the loop and a witness position $i'$ may evolve on $\sigma$.
Observe that there are three parts of the run iterating through the first loop.
In part one $\Phi$ holds nowhere because the balance (and hence the ratio) on the path to the (only) witness is insufficient.
It covers too many iterations of the bad loop.
In the last part, $\Phi$ holds everywhere because the ratio is sufficient.
In between it depends on local differences whether the ratio condition is satisfied or not.
The first and last part can be uniformly labelled and thus represented each by a copy of the original loop.
On the other hand, the intermediate part is short: its length depends only on the length of the loop and the ratio, more precisely, on the size of the denominator (3 in the example) as measure of how sensitive the property is to changes in the frequency on an arbitrarily long path.

\cref{lem:fltl-decomposition} formalises this observation.

\fltldecomposition*

\begin{proof}
Assume that $P$ is good for $\Phi$ and thus $\bal(P)>0$.
Consider the first (smallest) position $i\ge|u|$ on $\sigma=uv^nw$ where $(\sigma,i)\not\models\Phi$.
If $i$ does not exist or $i\ge|uv^{n-\hat{n}}|$ we can choose $n_1=n-\hat{n}-1$ and $n_2=n-\hat{n}-n_1 = 1$.

Otherwise let $n_1$ be the last iteration of $P$ entirely satisfying $\Phi$, that is such that $|uv^{n_1}| \le i < uv^{n_1+1}$, and $h=|uv^n_1|$.
Consequently we let $n_2=n-\hat{n}-n_1$.
Consider now any position $i'\in[h,h+|P|-1]$ in the $(n_1+1)$-th iteration where $\Phi$ still holds.
If there is none, then $\Phi$ does not hold in later iterations either and the statement of the lemma holds.

Since $(\sigma,i')\models\Phi$ there is some position $j>i'$ with $(\sigma,j)\models\psi$ and $\bal(\sigma(i')\sigma(i'+1)…\sigma(j-1))\ge0$
Observe that we can assume that $j\ge|uv^{n-1}|$ because otherwise $j+|P|$ would serve as witness since in that case
\begin{align*}
  \bal(\sigma(i')\sigma(i'+1)…\sigma(j-1+|P|))
    &= \bal(\sigma(i')\sigma(i'+1)…\sigma(j-1)) + \bal(P) \\
    &\ge \bal(\sigma(i')\sigma(i'+1)…\sigma(j-1))  \\
    &\ge0
\end{align*}
while $\sigma(j)=\sigma(j+|P|)$ and thus $\psi\in\labels_\mathcal{P}(\sigma(j+|P|))$.

However, the balance cannot be too large, more precisely, $\bal(\sigma(i')\sigma(i'+1)…\sigma(j-1)) \le |P|\cdot y$.
Depending on whether $i'<i$ or $i<i'$ we have
\[
  \bal(\sigma(i')…\sigma(j-1)) =
  \begin{cases}
    \bal(\sigma(i)…\sigma(j-1)) - \bal(\sigma(i)…\sigma(i'-1)) & \text{if $i<i'$}\\
    \bal(\sigma(i)…\sigma(j-1)) + \bal(\sigma(i')…\sigma(i-1)) & \text{if $i'<i$}\\
  \end{cases}
\]
Considering the first case, we can bound the difference by the maximal gain
\[
\bal(\sigma(i)…\sigma(i'-1))\le(|P|-1)\cdot(y-x)\le y|P|
\]
on a path of length at most $|P|-1$.
In the second case, the lower bound on the balance
\[
\bal(\sigma(i')…\sigma(i-1))\ge|P-1|\cdot(-x) \ge -y|P|
\]
is of interest because we conclude that in any case
\[
  \bal(\sigma(i')…\sigma(j-1)) \le \bal(\sigma(i)…\sigma(j-1)) + y|P| < y|P|
\]
Since $\bal(P)\ge1$ we have that
\begin{align*}
  \bal(\sigma(i'+\hat{n}|P|)…\sigma(j-1))
    & = \bal(\sigma(i')…\sigma(j-1)) - \hat{n}\cdot\bal(P) \\
    & < y|P| - y|P|\cdot\bal(P) \\
    & \le 0
\end{align*}
meaning that after at most $\hat{n}$ further iteration $\Phi$ can not hold any more.

Assuming now that $P$ is bad for $\Phi$ allows for similar reasoning.
Consider $i\ge|uv|$ to be the first position on $\sigma$ where $(\sigma,i)\models\Phi$ while $(\sigma,i-|P|)$
If $i$ does not exist or $i\ge|uv^{n-\hat{n}}|$ we can again choose $n_1=n-\hat{n}-1$ and $n_2=n-\hat{n}-n_1 = 1$.
Otherwise we choose $n_1$ such that $|uv^{n_1}|\le i<|uv^{n_1+1}|$ and let $h=|uv^{n_1}|$.

There is a position $j>i$ such that $\psi\in\labels_\mathcal{P}(\sigma(j))$ and
 $\bal(\sigma(i)…\sigma(j-1))\ge0$.
 Observe that $j\ge|uv^n|$ because otherwise $\bal(\sigma(i-|P|)…\sigma(j-1-|P|))\ge0$ and $\psi\in\labels_\mathcal{P}(\sigma(i-|P|))$ contradicting that $(\sigma,i-|P|)\not\models\Phi$.

Consider now any position $i'\in[h,h+|P|-1]$ where $(\sigma,i')\not\models\Phi$, if any.
We have
\[
  \bal(\sigma(i')…\sigma(j-1)) =
  \begin{cases}
    \bal(\sigma(i)…\sigma(j-1)) - \bal(\sigma(i)…\sigma(i'-1)) & \text{if $i<i'$}\\
    \bal(\sigma(i)…\sigma(j-1)) + \bal(\sigma(i')…\sigma(i-1)) & \text{if $i'<i$}
  \end{cases}
\]
and obtain the bounds
\[\begin{array}{rcccr@{\hspace{2cm}}c}
  \bal(\sigma(i)…\sigma(i'-1)) &\le & |P-1| \cdot(y-x) &\le &  y|P| & \text{(if $i<i'$)}\\
  \bal(\sigma(i')…\sigma(i-1)) &\ge & |P-1| \cdot(-x)  &\ge & -y|P| & \text{(if $i'<i$).}
\end{array}\]
Hence
\[
  \bal(\sigma(i')…\sigma(j-1)) \ge \bal(\sigma(i)…\sigma(j-1)) - y|P| \ge -y|P|
\]
Now, since $\bal(P)<0$ we have that
\begin{align*}
  \bal(\sigma(i'+\hat{n}|P|)…\sigma(j-1))
    & = \bal(\sigma(i')…\sigma(j-1)) - \hat{n}\cdot\bal(P) \\
    & \ge -y|P| - (y|P|\cdot\bal(P)) \\
    & \ge 0
\end{align*}
providing that after $\hat{n}=y|P|$ more iterations, $\Phi$ holds at every position on the loop.

If $P$ is neutral for $\Phi$ then an iteration of $P$ more or less does not change if there is a witness or not and $(\sigma,i)\models\Phi$ if and only if $(\sigma,i+|P|)\models\Phi$ for all $|u|\le i<|uv^{n-1}|$.
\end{proof}

\cref{lem:fltl-decomposition} provides a bound on how often we need to unfold $P_k$ at most in order to guarantee that $\sigma$ determines a unique labelling.
Recall we assumed that $\sigma$ repeats $P_k$ for $n\ge\hat{n}+2$ times.
In $\mathcal{P}$, we may hence replace $P_k$ by $(P_k,P'_k,…,P'_k,P_k)$ introducing two copies and a sequence of exactly $\hat{n}$ unfoldings of it.
The decomposition $\sigma=uv^{n_1}v^{\hat{n}}v^{n_2}w$ given by \cref{lem:fltl-decomposition}
provides a corresponding run $\sigma'$ of the obtained path schema and a unique labelling for all of the new components.
Now, we are only left with cases discussed earlier: two stable loops and $\hat{n}$ rows.
For each of the stable loops, we can estimate that establishing consistency requires no more than $7|\mathcal{K}|^2$ additional locations.
For each of the $\hat{n}$ new rows it no more than $|\mathcal{K}|^2$ additional locations.
We can conclude that $\mathcal{P}'$ can be constructed with in total no more than
\[
  |P_k| + \hat{n}|P_k| + 2\cdot7|\mathcal{K}|^2 + \hat{n}\cdot|\mathcal{K}|^2 \le 17y|\mathcal{K}|^3
\]
additional locations.
\end{proof}
\subsection{The Size of $\mathcal{P}_{\Phi}$}

The induction provides the construction of $\mathcal{P}_\Phi$ from $\mathcal{K}$ requiring (at most) one step for each subformula of $\Phi$.
Let $\mathcal{P}_0$ be the APS provided by the base case that covers all propositions occurring in $\Phi$.
As argued earlier, its size is bounded by $2|\mathcal{K}|$ and the length of every loop is bounded by $|\mathcal{K}|$.
Applying the induction step now recursively for $\Phi$, i.e., augmenting $\mathcal{P}_0$ consistently with more and more subformulae of $\Phi$ we obtain a sequence of possibly growing path schemas until $\mathcal{P}_\Phi$ is obtained after at most $|\sub(\Phi)|\le|\Phi|$ steps.

We have seen that in the case of a next formula, constructing the consistent schema $\mathcal{P}'$ from $\mathcal{P}$ requires at most one unfolding of some loop for each location in $\mathcal{P}$ and thus $|\mathcal{P}'|\le|\mathcal{P}|\cdot|\mathcal{K}|$.
In the case of an until formula \cref{lem:until-component} provides that for each component of $\mathcal{P}$ no more than $17y|\mathcal{K}|^3$ locations are added and thus $|\mathcal{P}'|\le|\mathcal{P}|\cdot17y|\mathcal{K}|^3$.
Counting the bits for representing $y$ to the length of $\Phi$ and hence estimating $y\le2^\Phi$ it follows that after $|\Phi|$ steps, the resulting path schema is of size
\[
  |\mathcal{P}_\Phi| \le |\mathcal{P}_0| \cdot(17\cdot 2^{|\Phi|}|\mathcal{K}|^3)^{|\Phi|} \in \mathcal{O}(2^{f(|\Phi|+|\mathcal{K}|)})
\]
for some polynomial $f$ and thus at most exponential in the size of the input.

By construction $\mathcal{P}_\Phi$ is correct and there is a run $\sigma\in\Runs(\mathcal{P}_\Phi)$ with $\st_{\mathcal{P}_\Phi}(\sigma)=\rho \models\Phi$ and hence $\labels(\mathcal{P}(0))=\labels_\mathcal{P}(\sigma(0))\ni\Phi$.
This completes the proof for \cref{lem:fltl-completeness}.

\end{document}